\documentclass[a4paper,10pt]{article}

\usepackage{dsfont}
\usepackage{appendix}
\usepackage{extsizes}
\usepackage{lipsum}
\usepackage[colorlinks]{hyperref}
\usepackage[auth-lg]{authblk}
\usepackage{tcolorbox}
\usepackage{kvoptions}
\usepackage{xparse}
\usepackage{color}
\usepackage{etoolbox}
\usepackage[framemethod=TikZ, xcolor={RGB,table}]{mdframed}
\usepackage{mathtools}

\usepackage{fancyhdr}
\usepackage{colonequals}
\usepackage{yfonts}
\usepackage{graphicx}
\usepackage{amstext,amssymb,amsmath,amsthm}
\usepackage[nice]{nicefrac}

\usepackage{mathrsfs}
\usepackage{caption}
\usepackage{centernot}
\usepackage{tikz}
\usetikzlibrary{arrows}
\usepackage{microtype}
\usepackage{algorithm,algpseudocode}
\usepackage[flushmargin,norule]{footmisc}
\usepackage[margin=2.5cm]{geometry}
\usepackage{dsfont}
\usepackage{tcolorbox}

\usetikzlibrary{automata,positioning}

\usepackage{thmtools}
\usepackage{thm-restate}
\usepackage{cleveref}
\usepackage{enumitem}
 \usepackage{optidef}

\tcbuselibrary{theorems}

\hypersetup{linkcolor=blue,citecolor=blue,urlcolor=blue}

\tcbuselibrary{listings,breakable}
\allowdisplaybreaks

\DeclareMathOperator*{\Exp}{\ensuremath{{\mathbb{E}}}}
\DeclareMathOperator*{\Prob}{\ensuremath{\textnormal{Pr}}}
\renewcommand{\Pr}{\Prob}

\newcommand{\prob}[1]{\Pr\paren{#1}}
\newcommand{\card}[1]{|#1|}
\newcommand{\paren}[1]{\left ( #1 \right )}
\newcommand{\bracket}[1]{\left [ #1 \right ]}
\newcommand{\expect}[1]{\Exp\bracket{#1}}
\newcommand{\var}[1]{\textnormal{Var}\bracket{#1}}
\newcommand{\set}[1]{\ensuremath{\left\{ #1 \right\}}}
\newcommand{\poly}{\mbox{\rm poly}}
\newcommand{\eps}{\varepsilon}

\renewcommand{\algorithmicrequire}{\textbf{Input:}}
\renewcommand{\algorithmicensure}{\textbf{Output:}}

\theoremstyle{definition}
\newtheorem{definition2}{Definition}
\theoremstyle{definition}
\newtheorem{example2}[definition2]{Example}
\theoremstyle{definition}
\newtheorem{property}[definition2]{Property}
\theoremstyle{definition}
\newtheorem{theorem2}[definition2]{Theorem}
\theoremstyle{definition}
\newtheorem{lemma2}[definition2]{Lemma}
\theoremstyle{definition}
\newtheorem{corollary2}[definition2]{Corollary}
\theoremstyle{definition}
\newtheorem{observation2}[definition2]{Observation}
\theoremstyle{definition}
\newtheorem{claim2}[definition2]{Claim}
\theoremstyle{definition}
\newtheorem{note2}[definition2]{Note}
\theoremstyle{definition}
\newtheorem{remark2}[definition2]{Remark}
\theoremstyle{definition}
\newtheorem{research2}[definition2]{Research Direction}
\theoremstyle{definition}
\newtheorem{conjecture2}[definition2]{Conjecture}
\theoremstyle{definition}
\newtheorem{proposition}[definition2]{Proposition}
\theoremstyle{definition}
\usepackage{enumitem}
\setenumerate[0]{label=(\alph*)} 

\newenvironment{tbox}{\begin{tcolorbox}[
		enlarge top by=3pt,
		enlarge bottom by=3pt,
		boxsep=0pt,
		left=4pt,
		right=4pt,
		top=10pt,
		arc=0pt,
		boxrule=1pt,toprule=1pt,
		colback=blue!2,
		colframe=blue,
		]%%
	}
{\end{tcolorbox}}

\newenvironment{tbox2}{\begin{tcolorbox}[
		enlarge top by=3pt,
		enlarge bottom by=3pt,
		breakable,
		boxsep=0pt,
		left=4pt,
		right=4pt,
		top=10pt,
		arc=0pt,
		boxrule=1pt,toprule=1pt,
		colback=blue!2,
		colframe=blue,
		]%%
	}
	{\end{tcolorbox}}

\newtheorem{mdalg}{Algorithm}

\title{Decremental Matching in General Graphs}
\author[1]{Sepehr Assadi\thanks{sepehr.assadi@rutgers.edu. Supported by NSF CAREER Grant CCF-2047061, and a gift from Google Research.}}{}
\author[1]{Aaron Bernstein\thanks{bernstei@gmail.com. Funded by NSF CAREER Grant 1942010.}}
\author[1]{Aditi Dudeja\thanks{aditi.dudeja@rutgers.edu}}
\affil{Rutgers University}
\date{}
\begin{document}
\maketitle
\begin{abstract}
We consider the problem of maintaining an approximate maximum integral matching in a dynamic graph $G$, while the adversary makes changes to the edges of the graph. The goal is to maintain a $(1+\eps)$-approximate maximum matching for constant $\eps>0$, while minimizing the update time. In the fully dynamic setting, where both edge insertion and deletions are allowed, Gupta and Peng (see \cite{GP13}) gave an algorithm for this problem with an update time of $O(\nicefrac{\sqrt{m}}{\eps^2})$.

Motivated by the fact that the $O_{\eps}(\sqrt{m})$ barrier is hard to overcome (see Henzinger, Krinninger, Nanongkai, and Saranurak \cite{HKNS15}; Kopelowitz, Pettie, and Porat \cite{KPP16}), we study this problem in the \emph{decremental} model, where the adversary is only allowed to delete edges. Recently, Bernstein, Probst-Gutenberg, and Saranurak (see \cite{BPT20}) gave an $O(\poly(\nicefrac{\log n}{\eps}))$ update time decremental algorithm for this problem in \emph{bipartite graphs}. However, beating $O(\sqrt{m})$ update time remained an open problem for \emph{general graphs}. 

In this paper, we bridge the gap between bipartite and general graphs, by giving an $O_{\eps}(\poly(\log n))$ update time algorithm that maintains a $(1+\eps)$-approximate maximum integral matching under adversarial deletions. Our algorithm is randomized, but works against an adaptive adversary. Together with the work of Grandoni, Leonardi, Sankowski, Schwiegelshohn, and Solomon \cite{GLSSS19} who give an $O_{\eps}(1)$ update time algorithm for general graphs in the \emph{incremental} (insertion-only) model, our result essentially completes the picture for partially dynamic matching. 
\end{abstract} 

\section{Introduction}

In dynamic graph algorithms, the main goal is to maintain a key property of the graph while an adversary makes changes to the edges of the graph. An algorithm is called \emph{incremental} if it handles only insertions, \emph{decremental} if it handles only deletions and \emph{fully dynamic} if it handles both insertions as well as deletions. The goal is to minimize the update time of the algorithm, which is the time taken by the algorithm to adapt to a single adversarial edge insertion or deletion and output accordingly. For incremental/decremental algorithms, one typically seeks to minimize the \emph{total update time}, which is the aggregate sum of update times over the \emph{entire} sequence of edge insertions/deletions.\\ \\
We consider the problem of maintaining a $(1+\eps)$-approximation to the maximum matching in a dynamic graph. In the fully dynamic setting, the best known update time for this problem is $O(\sqrt{m})$ (see \cite{GP13}), and the conditional lower bounds proved in the works of \cite{HKNS15} and \cite{KPP16} suggest that $O(\sqrt{m})$ is a hard barrier to break through. For this reason, several relaxations of this problem have been studied. For example, one line of research has shown that we can get considerably faster update times if we settle for large approximation factors (see for example \cite{BHI15,BHN16,BK21,BS16,Wajc2020,RSW22,BK22}). Another research direction has been to consider the more relaxed incremental or decremental models. In the incremental (insertion-only) setting, there have been a series of upper and lower bound results \cite{BLSZ14,Dahlgaard16,Gupta2014}, culminating in the result of \cite{GLSSS19}, who gave an optimal $O_{\eps}(m)$ \emph{total} update time (amortized $O_{\eps}(1)$) for $(1+\eps)$-approximate maximum matching. \\ \\
 The decremental (deletion-only) setting  requires an entirely different set of techniques. In fact, $O(\sqrt{m})$ update time ($O(m^{1.5})$ total time) remained the best known until recently, when \cite{BPT20} gave an $\poly(\nicefrac{\log n}{\eps})$ amortized update time algorithm for bipartite graphs. However, achieving a similar result for non-bipartite graphs remained an open problem. Our main theorem essentially closes the gap between bipartite and general graphs. 

\begin{restatable}{theorem2}{thmmain}\label{thm:mainintegral}	
		Let $G$ be an unweighted graph and let $\eps\in (0,\nicefrac{1}{2})$. Then, there exists a decremental algorithm with total update time $\Tilde{O}_{\eps}(m)$ (amortized $\Tilde{O}(1)$) that maintains an integral matching $M$ of value at least $(1-\eps)\cdot \mu(G)$, with high probability, where $G$ refers to the current version of the graph. The algorithm is randomized but works against an adaptive adversary. The dependence on $\eps$ is $2^{O(\nicefrac{1}{\eps^2})}$.
\end{restatable}
Our result largely completes the picture for partially dynamic matching by showing that in general graphs one can achieve $\poly(\log n)$ update time in both incremental and decremental settings. But there are a few secondary considerations that remain. Firstly, our update time is $O(\poly(\log n))$, rather than the $O(1)$ for the incremental setting (see \cite{GLSSS19}). Secondly, both the incremental of \cite{GLSSS19} and our decremental result for general graphs have an exponential dependence on $\nicefrac{1}{\eps}$, whereas incremental/decremental algorithms for bipartite graphs have a  polynomial dependence on $\nicefrac{1}{\eps}$ (see \cite{GLSSS19,Gupta2014}).  \\ \\
In algorithms literature, it has been the case that efficient matching algorithms for bipartite graphs do not easily extend to general graphs. Existence of blossoms (among other things), poses a technical challenge to obtaining analogous results for the general case. Consider the polynomial time algorithms for maximum matching for bipartite graphs, the most efficient algorithm, using alternating BFS, was discovered by Hopcroft and Karp; Karzanov (see \cite{HK73,Karzanov1973}) in 1973. However, several new structural facts and algorithmic insights were used by Micali and Vazirani to get the same runtime for general graphs (see \cite{MV80}). This is also a feature of recent work in different models, such as the streaming (see \cite{AG11,EKMS11} and \cite{MMU22}), fully dynamic (see \cite{BS15,BS16} and \cite{BHN16,BLM20}), and parallel models (see \cite{FGT16,ST17} and \cite{MV00,NV20}). We refer the reader to Section 1.3 of \cite{AV20} for a detailed discussion of this phenomenon.

\section{High-Level Overview}
\label{sec:overview}
Our algorithm for \Cref{thm:mainintegral} follows the high-level framework of \emph{congestion balancing} introduced by Bernstein, Probst-Gutenberg, and Saranurak \cite{BPT20}. They used this framework to solve approximate decremental matching in bipartite graphs, and also to solve more general flow problems. But the framework as they used it was entirely limited to cut/flow problems. As we discuss below, extending this framework to non-bipartite graphs introduces significant technical challenges. Moreover, our result shows how the key subroutine of congestion balancing is naturally amenable to a primal-dual analysis, which we hope can pave the way for this technique to be applied to other decremental problems.

\subsection{Previous Techniques} 

A key observation of \cite{BPT20} is that it is sufficient to develop an $\Tilde{O}_{\eps}(m)$ algorithm that does the following: it either maintains a fractional matching of size at least $(1-2 \eps)\cdot \mu(G)$ or certifies that $\mu(G)$ has dropped by a $(1-\eps)$ factor because of adversarial deletions. Since a result of~\cite{Wajc2020} enables us to round any bipartite fractional matching to an integral matching of almost the same value, their observation gives an algorithm to maintain an integral matching of size $(1-3\eps)\cdot \mu(G)$ in $\Tilde{O}_{\eps}(m)$ time under adversarial deletions. To motivate why \cite{BPT20} consider computing a fractional matching, consider the following ``lazy'' algorithm that works with an integral matching: compute an $(1+\eps)$ approximate integral matching $M$ of $G$ using a static $O(\nicefrac{m}{\eps})$ algorithm, wait for $\eps\cdot \mu(G)$ edges of $M$ to be deleted and then recompute the matching. Since we assume an adaptive adversary, the update time could be as large as $\Omega(\nicefrac{m^2}{\eps\cdot n})$, this is because the adversary could proceed by only deleting edges of $M$. As a result, the goal should be to maintain a \emph{robust} matching that can survive many deletions. Thus, \cite{BPT20} aim to maintain a ``balanced'' fractional matching $\vec{x}$ that attempts to put a low value on every edge. In doing so, the adversary will have to delete a lot of edges of $G$ to reduce the value of $\vec{x}$ by $\eps\cdot \mu(G)$. 

\paragraph*{Balanced Fractional Matching in Bipartite Graphs}

In order to ensure that the fractional matching is spread out and robust, the authors of \cite{BPT20} impose a capacity function $\kappa$ on the edges of the graph (initially, all edges have low capacity) and compute a fractional matching obeying these capacities. The main ingredient of the algorithm is the subroutine $\textsc{M-or-E*}(G,\eps,\kappa)$ which returns one of the following in $\Tilde{O}_{\eps}(m)$ time:
\begin{enumerate}
	\item A fractional matching $\vec{x}$ such that $\sum_{e\in E}x(e)\geqslant (1-\eps)\cdot \mu(G)$ and $x(e)\leqslant \kappa(e)$ for all $e\in E$, or,
	\item\label{item:one1} A set of edges $E^*$ such that have the following two properties. 
	\begin{enumerate}
	\item\label{item:onea} The total capacity through $E^*$ must be small: $\kappa(E^*)=O(\mu(G)\log n)$ and,
	\item\label{item:oneb} For all $\card{M}\geqslant (1-3\eps)\cdot \mu(G)$, $\card{M\cap E^*}\geqslant \eps\cdot \mu(G)$. 
	\end{enumerate}
\end{enumerate}
Property \ref{item:onea} ensures that the total capacity increase is small, while Property \ref{item:oneb} ensures that we only increase capacity on important edges that are actually needed to form a large matching. The authors of \cite{BPT20} show that $\textsc{M-or-E*}()$ can be used as a black-box to solve decremental matching: at each step, \textsc{M-or-E*}() is used to find a large fractional matching $\vec{x}$ (this matching is then rounded using \cite{Wajc2020} to get an integral matching), or to output the set $E^*$ along which we increase capacities. They are able to show that because of Properties \ref{item:onea} and \ref{item:oneb}, the edge capacities remain small on average.\\ \\
The \emph{congestion balancing} framework of \cite{BPT20} thus, consists of an outer algorithm that uses \textsc{M-or-E*}() as a subroutine. The outer algorithm for bipartite graphs carries over to general graphs as well. But, \textsc{M-or-E*}() is significantly more challenging to implement for the case of general graphs, so this subroutine will be our focus for the rest of the high level review. For the case of bipartite graphs the algorithm \textsc{M-or-E*}() is easier to implement because maximum fractional matchings correspond to maximum flows in bipartite graphs. Hence, existing algorithms for approximate maximum flows can be used to find the approximate maximum fractional matching obeying capacity $\kappa$. Moreover, if such a fractional matching is not large, then in bipartite graphs, the set of bottleneck edges is exactly a minimum cut of the graph. For general graphs, due to the odd set constraints, max flow, which was the key analytic and algorithmic tool in \cite{BPT20}, no longer corresponds to a maximum fractional matching that avoids the integrality gap.

\subsection{Our Contribution: Implementing \textsc{M-or-E*()} in General Graphs}

At a high-level, there are several structural and computational challenges to implementing \textsc{M-or-E*}() in the case of general graphs. We explain what the potential impediments are, and detail how our techniques circumvent these. 
\subparagraph*{Fractional Matchings in General Graphs}
In general graphs, not all fractional matchings have a large integral matching in their support and therefore, cannot be rounded to give a large matching. While fractional matchings that obey odd set constraints do avoid the integrality gap, it seems hard to compute such a matching that also obeys capacity function $\kappa$. In order to get past this, we define a candidate fractional matching that is both easy to compute as well as contains a large integral matching in its support. More concretely, our fractional matching either puts flow one through an edge, or a flow of value at most $\eps$. It can be proved that such a fractional matching obeys all small odd set constraints, and avoids the integrality gap. Our main contribution are two structural lemmas which show that we can find our candidate matching efficiently.
\begin{enumerate}
	\item First, given a graph $G$ with capacity $\kappa$, we want to determine if the \emph{value} of the maximum fractional matching obeying $\kappa$ and odd set constraints (denoted $\mu(G,\kappa)$) is at least $(1-\eps)\cdot \mu(G)$. In general graphs, we do this by giving a sampling theorem: let $G_s$ be the graph created by sampling edge $e$ with probability proportional to $\kappa(e)$, then $\mu(G_s)\geqslant \mu(G,\kappa)-\eps\cdot n$ with high probability. Thus, $\mu(G_s)$ is a good proxy for $\mu(G,\kappa)$ and it can be estimated efficiently by running any integral matching algorithm on $G_s$. 
\item 	Suppose we have determined at some point that $\mu(G,\kappa)$ is large, we are still left with the task of finding a fractional matching. Our next contribution is a structural theorem that enables us to deploy existing flow algorithms to find such a matching. Let $M$ be the approximate maximum matching of $G_s$. Let $M_{L}=\set{e\in M\mid \kappa(e)\leqslant \beta}$ and $M_{H}=\set{e\in M\mid \kappa(e)> \beta}$, where $\beta=O(\poly(\log n))$, and $L$ and $H$ are for low and high respectively. Let $V_{L}=V(M_L)$ and $V_{H}=V(M_H)$. Intuitively, $M$ breaks up our vertex set into two parts: vertices matched by low capacity edges (denoted $V_L$) and those that are matched by high capacity edges (denoted $V_{H}$). By adding some slack to our capacity constraints (we show that some slack can be incorporated in congestion balancing framework), we are able to treat the high-capacity edges as integral and compute a matching on $V_H$ using a black box for integral matching in general graphs. Additionally, we show that the maximum fractional matching on low capacity edges of $G[V_L]$ has value at least as much as $\card{M_L}$ (the low capacity edges of $M$), up to an additive error of $\eps\cdot n$. To compute this fractional matching, we show that because we are only considering edges of small capacity, small odd set constraints are automatically satisfied, so we can transform $G[V_L]$ into a \emph{bipartite graph} and then use an existing flow algorithm.
\end{enumerate} 
The second obstacle is finding the set $E^*$. As mentioned before, for the case of bipartite graphs, the max flow-min cut theorem gives us an easy characterization of the bottleneck edges. However, for the case of general graphs, this characterization is not clear. To get around this, we consider the dual of the matching LP of $G_s$, and show that the bottleneck edges can be identified by considering the dual constraints associated with the edges. This generalizes the cut-or-matching approach of \cite{BPT20}. Since $G_s$ is integral, we can compute the approximate dual by using an existing primal-dual algorithm of \cite{DP14} for integral matching in general graphs.\\ \\
Additionally, there are some secondary technical challenges as well. As mentioned before, our structural theorems only guarantee preservation of matching sizes up to an additive error of $\eps\cdot n$. When $\mu(G)=o(n)$, then the results, applied directly are insufficient for us. To get around this, we use a vertex sparsification technique to get $O_{\eps}(\log n)$ multigraphs which preserve all matchings of $G$, but contain only $O(\nicefrac{\mu(G)}{\eps})$ vertices. However, we now have to show that all of our ideas work for multigraphs as well. Finally, the rounding scheme of \cite{Wajc2020} cannot be applied as a black-box to any fractional matching in a general graph. Thus, unlike in \cite{BPT20}, we cannot use \cite{Wajc2020} as a black box and instead have to embed its techniques into the congestion balancing framework.  

\subparagraph*{Assumption on Matching Size}
Let $G$ be the input graph with vertex set $V$, note that if $\mu(G)\leqslant 100\log \card{V}$ at any time, then we can maintain a $(1+\eps)$ approximate matching using Definition 3.2 and Lemma 3.3 from \cite{GP13} to solve the problem in $\Tilde{O}(\nicefrac{m}{\eps})$ time. Thus, we only run our algorithm while $\mu(G)\geqslant 100\log \card{V}$. As mentioned before, our structural theorems, \Cref{lem:matchingGs} and \Cref{lem:sampling2} are proved for multigraphs $H$ with matching size at least $\mu(H)=\Omega(\eps\cdot \card{V(H)})$. This is because \Cref{lem:mainvertexred} allows us to reduce to the problem of decremental matching in multigraphs with large matching. Now, suppose $H$ is the multigraph obtained by running the reduction of \Cref{lem:mainvertexred} on $G$, the input graph. Then, the structural theorems, \Cref{lem:matchingGs} and \Cref{lem:sampling2} hold with probability at least $1-\exp(-\mu(H))$. Since $\mu(H)= \Omega(\mu(G))$ and $\mu(G)\geqslant 100\log \card{V}$, these structural theorems hold with high probability.

\section{Preliminaries}
\label{sec:prelimsmain}
We consider the problem of maintaining an approximate maximum integral matching in a graph $G$ in the decremental setting. In this setting, we are given a graph (possibly non-bipartite) $G_0=(V_0,E_0)$ with $\card{V_0}=n$ and $\card{E_0}=m$, and the adversary deletes edges from the graph one at a time. The goal is to maintain an approximate maximum matching of the graph $G$ as edges are deleted, while minimizing the \emph{total update time}, that is the aggregate sum of update times over the \emph{entire} sequence of deletions. 

\paragraph*{Notation.} Throughout the paper, we will use $G$ to refer to the current version of the graph, and let $V$ and $E$ be the vertex and edge sets of $G$ respectively. Additionally, let $\mu(G)$ denote the size of the maximum integral matching of $G$ (the current graph). During the course of the algorithm, we will maintain a fractional matching which corresponds to a non-negative vector $\vec{x}\in \mathbb{R}_{\geqslant 0}$ satisfying fractional matching constraints: $\sum_{e\ni v} x(e)\leqslant 1$. For a set $S\subseteq E$, then we let $x(S)=\sum_{e\in S} x(e)$. Given a capacity function $\kappa(e)$, we say that $x$ obeys $\kappa$ if $x(e) \leqslant \kappa(e)$ for all $e \in E$. For a vector $\vec{x}$, we use $\text{supp}(\vec{x})$ to be set of edges that are in the support of $\vec{x}$. For a fractional matching $\vec{x}$, we say that $\vec{x}$ satisfies odd set constraints if for every odd-sized $B\subseteq V$, $\sum_{e\in G[B]}x(e)\leqslant \frac{\card{B}-1}{2}$. 

Throughout this paper, for any $\eps\in (0,1)$, we will let $\alpha_{\eps}= \log n\cdot 2^{\nicefrac{60}{\eps^2}}$ and $\rho_{\eps}=\log n\cdot 2^{\nicefrac{40}{\eps^2}}$.  We now restate our main result.

\thmmain*

In order to solve this problem, we reduce to the case where it is sufficient to solve the same problem in multigraphs which have a large matching. More concretely, we prove the following theorem, and then show how that implies an algorithm for \Cref{thm:mainintegral}.

\begin{restatable}{theorem2}{mainfracthm}\label{thm:mainfrac}
	We are given an unweighted multi-graph $G_0=(V,E_0)$. Let $\card{V}=n$, and $\eps\in (0,\nicefrac{1}{2})$. Suppose $\mu(G_0)\geqslant  \nicefrac{\eps\cdot n}{100}$ and $\mu\geqslant (1-\eps)\cdot \mu(G)$. Suppose $G$ is subject to adversarial deletions. There is an algorithm \textsc{Dec-Matching}($G,\mu,\eps$), which processes deletions in $\Tilde{O}_{\eps}(m)$ total update time and has the following guarantees.
	\begin{enumerate}[label=(\alph*)]
		\item \label{item:onethm} When the algorithm terminates, $\mu(G)<(1-2\eps)\cdot \mu$. Upon termination the algorithm outputs ``\textsc{no}''.
		\item \label{item:twothm}Until the algorithm terminates, it maintains an integral matching with size at least $(1-20\eps)\cdot \mu$. 
	\end{enumerate}
\end{restatable}

The main contribution of our paper is to give an algorithm that proves \Cref{thm:mainfrac}. We first show how an algorithm for \Cref{thm:mainfrac} gives us an algorithm for \Cref{thm:mainintegral}. For this, we need the following reduction, which has been used previously in other settings such as for stochastic matchings (see for example \cite{AKL19} and \cite{CCEHMMV16}). Here, we use it in the decremental setting. 

\begin{restatable}{lemma2}{lemmreduc}\label{lem:reduction}\cite{AKL19}
	Let $\delta\in (0,1)$ and let $G$ be a graph with maximum matching size at least $\mu$ and at most $(1+\delta)\cdot \mu$. There is an algorithm \textsc{Vertex-Red}($G,\mu,\delta$) that takes as input $G$, $\mu$, and $\delta$, and in $O(m\cdot \nicefrac{\log n}{\delta^4})$ time returns $\lambda=\frac{100\cdot \log n}{\delta^4}$ multi-graphs $H_1,\cdots, H_{\lambda}$ that have the following properties, where the second property holds with probability at least $1-\exp\paren{\nicefrac{-\mu \cdot \log n}{\delta^4}}$.
	\begin{enumerate}[label=(\alph*)]
		\item \label{item:first} For all $i\in [\lambda]$, $\card{V(H_i)}= \frac{4\cdot (1+\delta)\cdot \mu}{\delta}$ and $E(H_i)\subseteq E(G)$. 
		\item\label{item:third} Suppose $G$ is subject to deletions, which we simulate on each of the $H_i$'s. More concretely, if an edge $e$ is deleted from $G$, then its copy in each of the $H_i$'s (if present) is also deleted. Suppose at the end of the deletion process, $\mu(G)\geqslant \tilde{\mu}$ for some $\tilde{\mu}\geqslant (1-2\delta)\cdot \mu$. Then, $\mu(H_i)\geqslant (1-\delta)\cdot \tilde{\mu}$ for some $i\in [\lambda]$. 
	\end{enumerate}
\end{restatable}

Additionally, we will need the following algorithm to compute an integral matching in the static setting in general graphs (see \cite{DP14}).

\begin{lemma2}\label{lem:static}\cite{DP14}
	There is an $O(\nicefrac{m}{\eps})$ time algorithm \textsc{Static-Match}() that takes as input a graph $G$, and returns an integral matching $M$ of $G$ with $\card{M}\geqslant (1-\eps)\cdot \mu(G)$. 
\end{lemma2}

We now show how \Cref{thm:mainfrac} in combination with \Cref{lem:reduction} and \Cref{lem:static} implies \Cref{thm:mainintegral}. We first give an informal description of the algorithm. 

\paragraph*{Description of Algorithm for Theorem \ref{thm:mainintegral} (\Cref{alg:proofthm1})} The algorithm takes as input a graph $G$, and a parameter $\eps>0$. It first estimates the size of the maximum matching of $G$ by running \textsc{Static-Match}($G,\eps$) on input $G$ and $\eps$. We denote this estimate by $\mu$. The algorithm then initiates a procedure to create multiple parallel instances of the algorithm \textsc{Dec-Matching}(). As mentioned in \Cref{thm:mainfrac}, the algorithm \textsc{Dec-Matching}() can maintain a $(1+\eps)$-approximate maximum matching in a multi-graph, provided the size of the matching in the multi-graph is large. So, we use the reduction mentioned in \Cref{lem:reduction} which creates multigraphs $H_1,H_2,\cdots, H_{\lambda}$, which have have $O(\nicefrac{\mu(G)}{\eps})$ vertices as opposed to $n$ vertices, and we run parallel instances of \textsc{Dec-Matching}($H_i,\mu\cdot (1-\eps),\eps$) for all $i\in [\lambda]$. This procedure also instantiates a pointer \textsc{cur}, which points to the least indexed $H_i$, with a matching of size at least $(1-\eps)\cdot \mu$. The algorithm will output the matching indexed by \textsc{cur} (denoted $M_{\textsc{cur}}$). The algorithm then moves on to a procedure for handling deletions of edges. If an edge $e\in G$ is deleted, then it is deleted from each of the $H_i$'s. If at this point, the size of the matching of any $H_\textsc{cur}$ by a large value, then we increment the value of \textsc{cur}. If the sizes of the maximum matchings of all $H_i$ have dropped by a significant amount, then we can conclude that with high probability, $\mu(G)$ has also dropped by a significant amount (by \Cref{lem:reduction}\ref{item:third}), and we restart the algorithm. 

%Then, it uses \textsc{Vertex-Red}() to get multi-graphs $H_1,\cdots,H_{\lambda}$ that have $O(\nicefrac{\mu(G)}{\eps})$ vertices as opposed to $n$ vertices. Moreover, each matching of $G$, which has size at least $(1-\eps)\cdot\mu(G)$, is approximately preserved in some $H_i$ (by the guarantees of \Cref{lem:reduction}\ref{item:third}). Thus, the algorithm \textsc{Dec-Matching}() is run on each of the $H_i$, and the deletions accumulated (stored in $W$) are fed to it. \textsc{Dec-Matching}() processes these deletions, until the maximum matching in $H_i$ is too small. At this point, we run \textsc{Dec-Matching}() on $H_{i+1}$, and we feed to it the updated set of deleted edges $W$. Finally, if the maximum matching size of each of the $H_i$'s is too small, then we can conclude that the maximum matching size of $G$ has also dropped by a significant amount (this is evident from \Cref{lem:reduction}\ref{item:third}). At this point, we re-estimate $\mu(G)$ by running \textsc{Static-Match}() on $G$, and $\eps$, and recompute $H_i$'s. 

\begin{algorithm}
	\algorithmicrequire{ Graph $G$ and a parameter $\eps>0$}\\
	\algorithmicensure{ A matching $M$ with $\card{M}\geqslant (1-8\eps)\cdot \mu(G)$}
	\caption{}
	\begin{algorithmic}[1]
		\State $M\leftarrow\textsc{Static-Match}(G,\eps)$. \label{item:line1}
		\State $\mu\leftarrow \card{M}$
		\Procedure{Instantiating \textsc{Dec-Matching}()}{}
			\State $i\leftarrow 1$.
			\State \textsc{cur}$\leftarrow \lambda +1$.
			\State Let $H_1,H_2,\cdots ,H_{\lambda}$ be the multi-graphs returned by $\textsc{Vertex-Red}(G,\mu,\eps)$. 
			\State Label $H_1,\cdots, H_{\lambda}$ as \textsc{active}.
			\For{$i\leqslant \lambda$}
			\If{$\card{\textsc{Static-Match}(H_i,\eps)}<(1-\eps)\cdot \mu$}
			\State Label $H_i$ as \textsc{inactive}. 
			\EndIf
			\EndFor
			\For{$H_i$ labelled \textsc{active}}
			\State Initialize $\textsc{Dec-Matching}(H_i,\mu\cdot (1-\eps),\eps)$ (denoted $\mathcal{A}_i$).
			\State Let $M_i$ be the matching maintained by $\mathcal{A}_i$. 
			\EndFor
			\State Let \textsc{cur} be the least index $i$ such that $H_i$ is \textsc{active}. 
			\If{\textsc{cur} $>\lambda$} 
			\State \textsc{terminate}.
			\EndIf
		\EndProcedure
		\Procedure{Handling deletion of edge $e$}{}
		\For{$H_i$ labelled \textsc{active}}
		\State Feed the deletion of $e$ to $\mathcal{A}_i$. \label{alg1:dec-matching}
		\If{$\mathcal{A}_i$ returns \textsc{no}}
		\State Label $H_i$ as \textsc{inactive}.
		\EndIf
		\EndFor
		\If{$H_{\textsc{cur}}$ labelled \textsc{active}}\label{line:curactive} \Comment{Label may changed in \Cref{alg1:dec-matching}}.
		\State Continue to output $M_{\textsc{cur}}$. 
		\Else
		\If{$\textsc{cur}>\lambda$}
			\label{line:restart}	\State Return to \Cref{item:line1}. \Comment{Start over with new estimate for $\mu$}.
		\Else
		\State $\textsc{cur}\leftarrow \textsc{cur}+1$
		\State Goto \Cref{line:curactive}.
		\EndIf
		\EndIf
		\EndProcedure	
	\end{algorithmic}
	\label{alg:proofthm1}
\end{algorithm}

\begin{proof}[Proof of \Cref{thm:mainintegral}]
	We show that \Cref{alg:proofthm1} proves \Cref{thm:mainintegral}. Suppose we run \Cref{alg:proofthm1} on $G,\eps$. We first argue that it returns an integral matching of $G$ of size at least $(1-10\eps)\cdot \mu(G)$, where $\mu\geqslant \mu(G)\cdot (1-\eps)$. To see this, first observe that the multi-graphs $H_1,\cdots, H_{\lambda}$ obtained as an output of \textsc{Vertex-Red}($G,\mu\cdot (1+\eps),\eps$) have $\frac{4\cdot (1+\eps)\cdot \mu}{\eps}\leqslant \nicefrac{8\cdot \mu}{\eps}$ vertices, and $\mu(H_i)\geqslant (1-\eps)\cdot \mu$. Thus, $H_i$ satisfy the requirements of \textsc{Dec-Matching}() mentioned in \Cref{thm:mainfrac}. Hence, each instance of $\textsc{Dec-Matching}(H_i,\mu,\eps)$, until it returns \textsc{no}, maintains a matching of size at least $(1-9\eps)\cdot \mu$, which is at least $(1-10\eps)\cdot \mu(G)$. \\ \\
	%When it terminates and returns \textsc{no}, then maximum matching size of all $H_i$ has dropped to $(1-3\eps)\cdot\mu$, otherwise it returns an integral matching of size at least $(1-8\eps)\cdot \mu$. Let \textsc{Vertex-Red}() be the algorithm associated with \Cref{lem:reduction}. We now show the correctness of \Cref{alg:proofthm1}. We first argue that the algorithm indeed returns an integral matching of size at least $(1-8\eps)\cdot \mu$ in $G$. First, consider subgraphs $H_i$ obtained from running \textsc{Vertex-Red}($G,\mu,\eps$). These graphs have $\nicefrac{4\cdot (1+\eps)\cdot \mu}{\eps}$ vertices (see \Cref{lem:reduction}\ref{item:first}). Moreover, we only run \textsc{Dec-Matching}() on a fixed $H_i$, if $\mu(H_i)\geqslant (1-\eps)\cdot \mu$ (see \Cref{alg:proofthm1} \Cref{line:check}), thus satisfying the conditions of \Cref{thm:mainfrac}. So, algorithm $\textsc{Dec-Matching}()$ when run on $H_i$ returns an integral matching of value at least $(1-8\eps)\cdot \mu\geqslant (1-9\eps)\cdot \mu(G)$. This corresponds to a matching in $G$ as well. \\ \\
	We now want to bound the runtime of the algorithm. To do this, we first make the following claim.
	\begin{claim2}\label{claim:restart}
		Each time \Cref{alg:proofthm1} executes \Cref{line:restart}, $\mu(G)$ has dropped by a factor of $(1-2\eps)$. 
	\end{claim2}
	\begin{proof}First observe that every time the algorithm returns to \Cref{item:line1}, it must be the case that all the \textsc{Dec-Matching}($H_i,\mu\cdot(1-\eps),\eps$) returned \textsc{no}, which implies by \Cref{thm:mainfrac}\ref{item:onethm} that for all $i\in [\lambda]$, $\mu(H_i) < (1-2\eps)\cdot (1-\eps)\cdot \mu$. This in turn implies that $\mu(G)<(1-2\eps)\cdot\mu$. Suppose for contradiction that this is not the case, then from \Cref{lem:reduction}\ref{item:third} we can conclude that there is a matching of size at least $(1-3\eps)\cdot \mu$ in some $H_i$. More specifically, if $\mu(G)\geqslant (1-2\eps)\cdot \mu$, then it satisfies the hypothesis of \Cref{lem:reduction}\ref{item:third}, and we know that $\mu(H_i)\geqslant (1-3\eps)\cdot\mu$ for some $i\in \bracket{\lambda}$. This implies that \textsc{Dec-Matching}($H_i,(1-\eps)\cdot \mu,\eps$) wouldn't have terminated for some $i\in \bracket{\lambda}$ (by \Cref{thm:mainfrac}\ref{item:onethm}). Consequently, \Cref{line:restart} wouldn't have been executed, which is a contradiction. Thus, each time we execute \Cref{line:restart}, $\mu(G)$ has dropped from at least $\mu$ to at most $(1-2\eps)\cdot \mu$. This proves our claim. 
	\end{proof}
	From Claim \ref{claim:restart}, we can conclude that we return to \Cref{line:recomp} at most $\log_{1+\eps} n$ times. Thus, to upper bound the runtime of \Cref{alg:proofthm1}, it is sufficient to upper bound the runtime of each of the procedures. The runtime of the first procedure, which instantiates $\lambda$ instances of \textsc{Dec-Matching}($H_i,\mu\cdot(1-\eps),\eps$) is dominated by the runtime of \textsc{Vertex-Red}($G,\mu,\eps$), which is run once in the procedure, and the runtime of \textsc{Static-Match}($H_i,\eps$) which is run $\lambda$ times (once per multi-graph $H_i$). Thus, the total time of this procedure is $O(\nicefrac{m}{\eps}\cdot \log\nicefrac{1}{\eps}\cdot \lambda)$, which is $O(\nicefrac{m}{\eps^5}\cdot \log n\cdot \log \nicefrac{1}{\eps})$.\\ \\
	Next, we upper bound the runtime of the second procedure. If an edge $e$ is deleted from $G$, then the time to delete it from each of the multigraphs $H_1,\cdots, H_{\lambda}$ is $\lambda$. Finally, we are also maintaining at most $\lambda$ active instances of $\mathcal{A}_i$, and their total update time is $O_{\eps}(m\cdot \poly(\log n))$. Consequently, we have that the total update time of \Cref{alg:proofthm1} is $O_{\eps}(m\cdot \poly(\log n))$. 
\end{proof}
\section{Algorithm for \Cref{thm:mainfrac}}

We use $\mu(G,\kappa)$ to denote the value of the maximum fractional matching of $G$ obeying the capacity function $\kappa$ and the odd set constraints. To give an algorithm for \Cref{thm:mainfrac}, we need two ingredients. We first maintain a balanced fractional matching of the multigraph. To do this, we will show an algorithm \textsc{M-or-E*}() that given as input a capacitated graph $G$ either outputs a fractional matching nearly obeying these capacities if $\mu(G,\kappa)\geqslant (1-\eps)\cdot \mu(G)$, or if $\mu(G,\kappa)<(1-\eps)\cdot \mu(G)$, it outputs a set of edges $E^*$ along which capacities must be increased. When we say that the capacities are nearly obeyed, we mean that the can be exceeded only by a small factor. Then, we round the fractional matching using a known algorithm. We start with some definitions, and then go on to state the guarantees of the two algorithms, and show how it gives us an algorithm for \Cref{thm:mainfrac}.

\begin{definition2}
	Given a multi-graph $G$, for a pair of vertices $u$ and $v$, define $D(u,v)$ to be the set of edges between $u$ and $v$. Similarly, if $e$ is an edge between $u,v$, then $D(e)\coloneqq D(u,v)$. \end{definition2}

 \begin{definition2}\label{def:dist}
 	Let $G$ be a multigraph with $n$ vertices and $m$ edges. Let $\kappa$ be a capacity function on the edges. Suppose $\vec{x}$ is a fractional matching of $G$ ($\vec{x}$ is a vector of length $m$). Then, we define $\vec{x}^C$ to be a vector of support size at most $\min\set{m,{n\choose 2}}$, where for a pair of vertices, $u$ and $v$, $x^C{(u,v)}\coloneqq\sum_{e\in D(u,v)}x(e)$. Essentially, if $\vec{x}$ is a fractional matching on a multigraph, then $\vec{x}^C$ is a fractional matching obtained by ``collapsing'' all edges together. Similarly, if $\vec{y}$ is a vector of support at most  $n\choose 2$, then we define $\vec{y}^{D}$ to be an $m$ length vector such that for every $e\in E$ between a pair of vertices $u$ and $v$, $y^D(e)\coloneqq\frac{y(u,v)\cdot \kappa(e)}{\kappa(D(e))}$ ($D$ is for distributed, and we distribute the flow among the edges in proportion to their capacity). 
 \end{definition2}
 
 \begin{remark2}
 	Note that in doing the transformations in \Cref{def:dist}, the support size of the transformed vector is always at most $m$. Thus, it doesn't negatively affect our runtime. 
 \end{remark2}
 
 We now state our core ingredient, which either finds a balanced fractional matching of the multigraph $G$, or gives a set of edges $E^*$ along which we can increase capacity.

\begin{restatable}{lemma2}{lemtrio}\label{lem:MorE}
	Let $G$ be a multi-graph with $\mu(G)\geqslant \nicefrac{\eps\cdot n}{16}$. Let $\kappa$ be a capacity function on the edges of $G$. There is an algorithm \textsc{M-or-E*}(), that takes as input $G,\kappa, \eps\in (0,\nicefrac{1}{2})$ and $\mu\geqslant (1-\eps)\cdot \mu(G)$ and in time $O(\nicefrac{m\cdot \log n}{\eps})$ returns one of the following. 
	\begin{enumerate}[label=(\alph*)]
		\item\label{item:aa} A fractional matching $\vec{x}$ of value at least $(1-10\eps)\cdot \mu$ with the following properties.
		\begin{enumerate}[label=(\roman*)]
			\item\label{item:ai} For any $e\in \text{supp}(\vec{x})$ with $\kappa(D(e)) > \nicefrac{1}{\alpha_{\eps}^2}$, $x(e)=\frac{\kappa(e)}{\kappa(D(e))}$, and $x(D(e))=1$.
			\item\label{item:aii} For any $e\in \text{supp}(\vec{x})$ with $\kappa(D(e)) \leqslant \nicefrac{1}{\alpha_{\eps}^2}$, $x(e)\leqslant \kappa(e)\cdot \alpha_{\eps}$ and $x(D(e))\leqslant \kappa(D(e))\cdot \alpha_{\eps}$. 
		\end{enumerate}
		\item\label{item:bi} \label{item:edgeE} A set $E^*$ of edges such that $\kappa(E^*)=O(\mu \log n)$ such that for any integral matching $M$ with $\card{M}\geqslant (1-3\eps)\cdot \mu$, we have $\card{M\cap E^*}\geqslant \eps \mu$. Moreover, $\kappa(e)<1$ for all $e\in E^*$. Additionally, for every pair of vertices $u,v\in V$, either $D(u,v)\cap E^*=\emptyset$ or $D(u,v)\subseteq E^*$.
	\end{enumerate}
\end{restatable}

We give some intuition for \Cref{lem:MorE}. Recall that we need a balanced fractional matching that contains a large integral matching in its support. However, as mentioned before, in the case of general graphs, finding such a balanced fractional matching is not straightforward. In order to get past this obstacle, we define a balanced fractional matching that is easy to find and also avoids the integrality gap. We will explain how to find it in the subsequent sections. For now, we explain at a high-level, why the fractional matching $\vec{x}$ found by \Cref{lem:MorE} avoids the integrality gap. Consider $\vec{x}^C$. Observe from \Cref{lem:MorE}\ref{item:aa}, that for any pair of vertices $u,v\in G$, either $x^C((u,v))=1$ or $x^C((u,v))\leqslant \eps$. Thus, $\vec{x}$ satisfies odd-set constraints for all odd sets of size at most $\nicefrac{1}{\eps}$. By a folklore lemma, we can then argue that $\vec{x}$ contains an integral matching of size at least $(1+\eps)^{-1}\cdot \sum_{u\neq v}x^C((u,v))$. \\ \\
We will use \textsc{M-or-E*}() as a subroutine in the algorithm \textsc{Dec-Matching}(). The fractional matching output by \textsc{M-or-E*}() will have certain properties, we state these properties now, since it will be helpful in visualizing the fractional matching. We give a proof for these later on.  

\begin{property}
We will use \textsc{M-or-E*}($G,\mu,\kappa,\eps$) as a subroutine in \textsc{Dec-Matching}($G,\mu,\eps$) to get a matching $\vec{x}$ with the following properties.
\begin{enumerate}
\item Each time \textsc{M-or-E*}() returns the set $E^*$, we increase capacity along $E^*$ by multiplying $\kappa(e)$ for each $e\in E^*$ by the same factor. 
\item Consider $u,v\in V$ and let $e$ and $e'$ be edges in between $u,v$. Then, $\kappa(e')=\kappa(e)$ at all times during the run of the algorithm.
\end{enumerate} 
\end{property}

The next property follows immediately from \Cref{lem:MorE}\ref{item:a}.
\begin{property}\label{prop:capexceed}
	Let $\vec{x}$ be the matching output by \textsc{M-or-E*}($G,\mu,\kappa,\eps$), then $x(e)\leqslant \kappa(e)\cdot \alpha_{\eps}^2$ for all $e\in E$. 
\end{property}

\begin{definition2}
	Let $G$ be a multi-graph, let $\kappa$ be a capacity function on the edges of $G$ and let $\vec{x}$ be a fractional matching obeying $\kappa$. Let $\eps\in (0,1)$. Then, we split $\vec{x}$ into two parts, $\vec{x}^f$ and $\vec{x}^i$, where $\vec{x}=\vec{x}^f+\vec{x}^i$, and $\text{supp}(\vec{x}^f)=\set{e\in E\mid \kappa(D(e))< \nicefrac{1}{\alpha_{\eps}^2}}$ and $\text{supp}(\vec{x}^i)=\set{e\in E\mid \kappa(D(e))\geqslant \nicefrac{1}{\alpha_{\eps}^2}}$ (here, $\vec{x}^f$ stands for fractional and $\vec{x}^i$ stands for integral. Although the edges in $\vec{x}^i$ are not technically integral, they are large enough for us to round them).
\end{definition2}

We briefly state the implications of this definition. We do not use these properties in our proof, but it will be useful to state it nevertheless.

\begin{property}\label{prop:disjsupp}
	Let $G$ be any multigraph, and let $\vec{x}$ be a fractional matching of $G$. Then, for any pair of vertices $u,v$, either $D(u,v)\subseteq \text{supp}(\vec{x}^i)$ and $D(u,v)\cap \text{supp}(\vec{x}^f)=\emptyset$ or, $D(u,v)\subseteq \text{supp}(\vec{x}^f)$ and $D(u,v)\cap \text{supp}(\vec{x}^i)=\emptyset$.
\end{property}

\begin{observation2}\label{obs:suppmatch}
	Suppose $\vec{x}$ is a fractional matching returned by \textsc{M-or-E*}() and consider $\vec{z}=\vec{x}^i$. Then supp($\vec{z}^C$) is a matching. This is implied by \Cref{lem:MorE}\ref{item:a}\ref{item:ai}.
\end{observation2}

The second ingredient is a method to round fractional matchings. This theorem is implicit in \cite{Wajc2020}, however, we prove it in \Cref{sec:rounding} for completeness.

\begin{restatable}{lemma2}{lemuno}\label{thm:sparse}\cite{Wajc2020}
	Suppose $G$ is an unweighted simple graph, let $\eps\in(0,\nicefrac{1}{2})$, and let $\vec{x}$ be a fractional matching of $G$ such that $x(e)\leqslant \eps^6$. Then, there is a dynamic algorithm \textsc{Sparsification}($\vec{x},\eps$), that has the following properties.
	\begin{enumerate}[label=(\alph*)]
		\item \label{lem:sparseb}  The algorithm maintains a subgraph $H\subseteq\text{supp}(\vec{x})$ such that $\card{E(H)}=O_{\eps}(\mu(\text{supp}(\vec{x}))\cdot \poly(\log n))$, and with high probability $\mu(H)\geqslant (1-\eps)\cdot \sum_{e\in E}x(e)$.
		\item \label{lem:sparsea} The algorithm handles the following updates to $\vec{x}$: the adversary can either remove an edge from supp($\vec{x}$) or for any edge $e$, the adversary can reduce $x(e)$ to some new value $x(e)\geqslant 0$. 
		\item \label{lem:sparsec} The algorithm handles the above-mentioned updates in $\Tilde{O}_{\eps}(1)$ worst-case time. 
	\end{enumerate}
\end{restatable}

We now show how \textsc{M-or-E*}() and \textsc{Sparsification}() give us an algorithm for \Cref{thm:mainfrac}. This algorithm, called \textsc{Dec-Matching}() is stated in \Cref{alg:decmatching}. Before proving \Cref{thm:mainfrac}, we give a brief description of \Cref{alg:decmatching}.
\paragraph*{Description of \textsc{Dec-Matching}().} The algorithm takes as input a multigraph $H$, a parameter $\eps$, an estimate $\mu$ on the size of the maximum matching of $H$. It then instantiates the capacities of the graph $H$. The algorithm then starts a new phase by running  \textsc{M-or-E*}() on this capacitated graph. The algorithm returns either a large fractional matching or if the matching is small, then it outputs a set of edges $E^*$. The algorithm \textsc{Dec-Matching}() in the latter case, increases capacities along $E^*$. Suppose at some point the algorithm finds a large fractional matching $\vec{x}$. Recall that $\vec{x}$ output by \textsc{M-or-E*}() is a matching with the following property: consider $\vec{x}^C$, then either $x^C((u,v))=1$ or $x^C((u,v))\leqslant \nicefrac{1}{\alpha_\eps}$ (this is guaranteed by \Cref{lem:MorE}\ref{item:aa}). The algorithm extracts the latter part, and sparsifies it using \textsc{Sparsification}(). Since the fractional matching input to \textsc{Sparsification}() has flow at most $\nicefrac{1}{\alpha_{\eps}}$ between any pair of vertices, it satisfies the hypothesis of \Cref{thm:sparse} and obtain a graph $S$ containing $\Tilde{O}(\mu(H))$ edges and a large integral matching. Thus, \textsc{Stat-Match}($S,\eps$) is run and this matching is combined with the ``integral'' part of $\vec{x}$ and is output. 
%the algorithm sparsifies it to compute a subgraph $S$, which contains a large integral matching. The fractional matching obtained has the property that it avoids the integrality gap. Thus, the integral matching in $S$ has roughly the same size as the value of the fractional matching. The static approximate matching is then run on $S$ to compute a large integral matching, $M$. The algorithm then deals with the deletions. If at some point, at least $\eps\mu$ edges of the matching $M$ are deleted, then the algorithm recomputes the sparsifier $S$. Since the update time of the sparsification algorithm is $\Tilde{O}(1)$ in the worst case, the total update time for maintaining $S$ over all $m$ deletions is at most $\Tilde{O}(m)$. Moreover, $S$ has $\Tilde{O}(\mu(H))$ edges, and so we can quickly compute matching in $S$ (since runtime of \textsc{Static-Match}() will be $\Tilde{O}(\mu(H))$). The algorithm \textsc{Dec-Matching}() also keeps a tab on the deletions to the fractional matching $\vec{x}$. If at some point the value of $\vec{x}$ drops by at least $\eps\mu$, the algorithm recomputes $\mu(H)$. If $\mu(H)$ has dropped significantly, it terminates the algorithm. If $\mu(H)$ is still large, then it proceeds to recompute $\vec{x}$.

\begin{algorithm}[H]
	\caption{\textsc{Dec-Matching}($H,\mu,\eps$)}
	\label{alg:decmatching}
	\begin{algorithmic}[1]
		\State \label{line:initcap}For all $e\in E(H)$, $\kappa(e)\leftarrow \nicefrac{1}{\alpha_{\eps}^{\lceil \log_{\alpha_{\eps}}n \rceil}}$
		\Procedure{Starting a new phase}{} 
		\State $\mu'\leftarrow$ \textsc{Static-Match}($H,\eps$)\label{line:recomp}\Comment{$\mu'\geqslant (1-\eps)\cdot \mu(H)$}
		\If{$\mu'\leqslant (1-3\eps)\cdot \mu$} 
		\State Return \textsc{no}, and terminate.\label{line:funcinfty}
		\Else 
		\While{\textsc{M-or-E*}($H,\kappa,\eps,\mu'$) returns $E^*$}
		\State \label{line:capchange} For all $e\in E^*$, $\kappa(e)\leftarrow \kappa(e)\cdot \alpha_{\eps}$
		\EndWhile
		\EndIf
		\EndProcedure
		\State Let $\vec{x}$ be the matching returned by \textsc{M-or-E*}(). \Comment{This is a matching in a multigraph.}
		\State $\vec{y}\leftarrow \vec{x}^f$, $\vec{z}\leftarrow \vec{x}^i$ \label{line:recomp2}\Comment{We will update $\vec{y},\vec{z}$ as edges are deleted.}\label{line:yz}
		\State $S\leftarrow \textsc{Sparsification}(\vec{y}^C,\eps)$ \Comment{Recall that Sparsification is dynamic} \label{line:sparsification}
		\State \label{line:rounding}$M\leftarrow \textsc{Static-Match}(S,\eps)\cup \text{supp}(\vec{z}^C)$. \Comment{The matching $M$ that is output.}
		\State $\textsc{CounterM}\leftarrow 0$ \Comment{Counter for deletions to the integral matching.}
		\State $\textsc{Counterx}\leftarrow 0$ \Comment{Counter for deletions to the fractional matching.}
		\Procedure{For deletion of edge $e$ from $H$}{}
		\If{$e\in \text{supp}(\vec{x})$} 
		\State Delete $e$ from $\text{supp}(\vec{x})$
		\State Update $\vec{z}^C$ accordingly. \Comment{See \Cref{line:yz}.}
		\State \label{line:sparsification-update}Update $\vec{y}^C$ accordingly; \textsc{Sparsification} from Line \ref{line:sparsification} then updates $S$.
		\State \textsc{Counterx} $\leftarrow$ \textsc{Counterx} $+\vec{x}(e)$.
		\EndIf
		\If{\textsc{Counterx} $>\eps\cdot \mu$}
		\State Goto \Cref{line:recomp}  \Comment{End current phase and start new one.}
		\EndIf
		\If{$e\in M$}
		\State Delete $e$ from $M$
		\State \textsc{CounterM} $\leftarrow$ \textsc{CounterM} $+1$ 
		\EndIf
		\If{$\textsc{CounterM}>\eps\cdot \mu$} \Comment{The matching $M$ is out of date}
		\State $\textsc{CounterM}\leftarrow 0$
		%\State $S\leftarrow \textsc{Sparsification}(\vec{y}^C,\eps)$
		\State \label{line:rebuild} $M\leftarrow \textsc{Static-Match}(S,\eps)\cup \text{supp}(\vec{z}^C)$ \Comment{Recompute outputted matching $M$}
		\EndIf
		\EndProcedure
	\end{algorithmic}
\end{algorithm}
The algorithm then begins a procedure to handle deletions. If an edge $e$ is deleted, then it is removed from \text{supp}($\vec{x}$), and \textsc{Sparsification}() algorithm updates $S$ in $\Tilde{O}_{\eps}(1)$ worst case time (\Cref{thm:sparse}\ref{lem:sparsec}). The algorithm also maintains counters for deletions to $\vec{x}$ and $M$. If the total value deleted from $\vec{x}$ exceeds $\eps\cdot \mu$, then it ends the current phase, and starts a new one. On the other hand, if the number of edges deleted from $M$ exceeds $\eps\cdot \mu$, then, the algorithm recomputes $M$ by running \textsc{Static-Match}($S,\eps$) again (since we know that $\vec{x}$ still contains a large integral matching in its support). \\ \\
We also state the following lemma which we prove later in \Cref{sec:lemma7}.

\begin{lemma2}\label{lem:runtime}
	The subroutine \textsc{M-or-E*}() and \Cref{line:recomp} are called at most $O(\alpha_{\eps}^3\log n)$ times in \Cref{alg:decmatching} during the course of deletion of $m$ edges. 
\end{lemma2}

We restate our theorem and give a proof. 

\mainfracthm*

\begin{proof}[Proof of \Cref{thm:mainfrac}]
	We first argue that \Cref{alg:decmatching} satisfies the properties of \Cref{thm:mainfrac}. First observe that when the algorithm executes \Cref{line:funcinfty}, $\mu'\leqslant (1-3\eps)\cdot \mu$. Since $\mu'\geqslant (1-\eps)\cdot\mu(H)$, this implies that $\mu(H)\leqslant (1-2\eps)\cdot \mu$. Thus, when the algorithm terminates, $\mu(H)\leqslant (1-2\eps)\cdot \mu$, which proves \Cref{thm:mainfrac}\ref{item:onethm}. \\ \\
	Next, we want to argue that while the algorithm does not terminate, it outputs an integral matching of size at least $\mu\cdot(1-20\eps)$. This corresponds to \Cref{thm:mainfrac}\ref{item:twothm}. We first argue that the algorithm indeed outputs matching. Note that the output matching $M$ is the union of \text{supp}($\vec{z}^C$) and the output of \textsc{Sparsification}($\vec{y}^C,\eps$). Note that $M$ is a union of two matchings, is implied by Observation \ref{obs:suppmatch} and \Cref{thm:sparse}. Moreover, these matchings are vertex disjoint and this follows from Property \ref{prop:disjsupp}. Thus, $M$ is a matching. \\ \\
	Next,we want to argue that $\card{M}\geqslant (1-20\eps)\cdot \mu$. First, observe that the fractional matching output by \textsc{M-or-E*}() has value at least $(1-10\eps)\cdot\mu$. Next, observe that because of \Cref{lem:MorE}\ref{item:a}\ref{item:aii}, the matching output by \textsc{M-or-E*}() satisfies the conditions of \Cref{thm:sparse}: that is, $\vec{y}^C$ has the property that the flow through any edge in the support, $\vec{y}^C(e) \leqslant \nicefrac{1}{\alpha_{\eps}}\leqslant \eps^3$. Thus, $\mu(S)\geqslant (1-\eps)\cdot \sum_{e\in \text{supp}(\vec{y}^C)}y^C(e)$. Therefore, the matching output by \textsc{Sparsification}($\vec{y}^C,\eps$) is an integral matching of size at least $(1-2\eps)\cdot \sum_{e\in \text{supp}(\vec{y}^C)}\vec{y}^C(e)$. Thus, we can conclude,
	\begin{align*}
	\card{M}&	\geqslant \sum_{e\in\text{supp}(\vec{x}^i)}x(e)+(1-2\eps)\cdot\sum_{e\in\text{supp}(\vec{x}^f)}x(e)\geqslant (1-2\eps)\cdot \sum_{e\in \text{supp}(\vec{x})}x(e)\\
	&\geqslant (1-2\eps)\cdot (1-10\eps)\cdot \mu\geqslant (1-12\eps)\cdot \mu.
	\end{align*}
Thus, the algorithm maintains an integral matching of size at least $(1-20\eps)\cdot \mu$ at all times. \\ \\ 
	{\bf Running Time:} We now bound the runtime of the algorithm. Let us first consider the time to initialize a phase. From  \Cref{lem:runtime} we know that  \textsc{M-or-E*}() is called $O(\alpha_{\eps}^3\log n)$ times, so this is also an upper bound on the number of phases. Thus, \textsc{Static-Match}() in Lines \ref{line:recomp} and \ref{line:rounding} is also called at most $O(\alpha_{\eps}^3\log n)$ times, since these are only called once per phase. All of these subroutines have runtime $\tilde{O}_{\eps}(m)$, so since each is executed $\tilde{O}_{\eps}(1)$ times, the total runtime of all of them is $\tilde{O}_{\eps}(m)$.\\ \\
	%The runtime of the algorithm is dominated by the number of times \textsc{M-or-E*}() and \textsc{Static-Match}() are called. From \Cref{lem:runtime}, we can conclude that \textsc{M-or-E*}() is called $O(\alpha_{\eps}^3\log n)$ times, and so the same goes for \textsc{Static-Match}() in Lines \ref{line:recomp} and \ref{line:rounding}, since these are only done once er hase. So, the total runtime of these subroutines is $\Tilde{O}(m)$. Next, we also need to bound the runtime of \textsc{Sparsification} and the total contribution of \textsc{Static-Match}() in \Cref{line:rebuild}.
	We also need to bound the runtime of \textsc{Sparsification}($\vec{y}^C,\eps$) and the total contribution of \textsc{Static-Match}() in \Cref{line:rebuild}. Let us start by bounding the runtime of \textsc{Sparsification}(). The time to initialize this procedure in \Cref{line:sparsification} is $\tilde{O}(m)$, and this is only done once per phase, so again the total time spent on initialization is $\tilde{O}(m)$. We also spend time updating the output sparsifier $S$ in Line \ref{line:sparsification-update}.
	The worst-case runtime time of  \textsc{Sparsification}() is $O(\log n)$ per update to $\vec{x}$ (see \Cref{thm:sparse}\ref{lem:sparsea} and \ref{lem:sparsec}). But each update to $\vec{x}$ corresponds to an adversarial deletion of some edge $e$, so there at most $m$ adversarial deletions, so the total overall time (among all phases) to update \textsc{Sparsification} is $\tilde{O}(m)$. \\ \\
	%This subroutine is called in two different situations.
	%\begin{enumerate}[label=(\alph*)]
	%	\item When \textsc{Counterx} has value has increased to $\eps\cdot \mu$ from $0$, and a new phase has started. In this case, we recompute the fractional matching $\vec{x}$, and run \textsc{Sparsification}() on $\vec{y}^C$, where $\vec{y}=\vec{x}^f$. Since the $\Tilde{O}(m)$ is the time taken to compute this matching (see \Cref{lem:MorE}), the time taken by the \textsc{Sparsification}() in this case is $\Tilde{O}(m)$ (since at most $t$ updates can be made in update time $t$) Since the total number of phases is $O(\log n)$, the total contribution of \textsc{Sparsification}() to the runtime when it is run in this situation is at most $\Tilde{O}(m)$.
	%\item When \textsc{CounterM} has value has increased to $\eps\cdot \mu$ from $0$. In this case, if $t$ updates have been made to $\vec{x}$. Note that updates to $\vec{x}$ are just adversarial deletions. These correspond to updates of the type \Cref{thm:sparse}\ref{lem:sparsec} to $\vec{y}^C$, which is the input to \textsc{Sparsification}(). Thus, the total update time of \textsc{Sparsification}() is $\Tilde{O}(t)$. Since the total number of updates is at most $m$, the total contribution of \textsc{Sparsification}() when it is called in this situation is at most $\Tilde{O}(m)$. 
	%\end{enumerate}
	Finally, we want to count the contribution of \textsc{Static-Match}() in \Cref{line:rebuild}. Observe that by  \Cref{thm:sparse}, the subgraph $S$ has size $\Tilde{O}(\mu)$. Thus, the runtime of \textsc{Static-Match}($S,\eps$) is $\Tilde{O}(\nicefrac{\mu}{\eps})$. Now, observe that \Cref{line:rebuild} is only executed when \textsc{CounterM} has increased from $0$ to $\eps\cdot \mu$, and the counter increases by at most $1$ per adversarial deletion, so there are at least $\eps\cdot \mu$ adversarial deletions per execution of \Cref{line:rebuild}. Thus, \textsc{Stat-Match}($S,\eps$) is called at most $\nicefrac{m}{\eps\cdot \mu}$ times in \Cref{line:rebuild}, and the total contribution here is $O(\nicefrac{m}{\eps^2})$.
	
	%Moreover, this is also run in two situations. 
	
	%\begin{enumerate}[label=(\alph*)]
	%	\item When a new phase starts: in this case, since the total number of phases is $O(\log n)$ (or equivalently, \textsc{Counterx} has increased from $0$ to $\eps\cdot\mu$), the total contribution of $\textsc{Stat-Match}(S,\eps)$ of this type, across all phases is $\Tilde{O}(\nicefrac{\mu}{\eps})$. 
	%	\item When \textsc{CounterM} has increased from $0$ to $\eps\cdot \mu$. In this case, $\eps\cdot \mu$ deletions have been performed. Thus, \textsc{Stat-Match}($S,\eps$) is called at most $\nicefrac{m}{\eps\cdot \mu}$ times in this case, and the total contribution here is $O(\nicefrac{m}{\eps^2})$. 
	%\end{enumerate}
	%This concludes our proof.
\end{proof}
\section{Proof of \Cref{lem:runtime}}
\label{sec:lemma7}

To prove \Cref{lem:runtime}, which gives an upper bound on the number of times the subroutine \textsc{M-or-E*}() is called in \Cref{alg:decmatching}, we will follow the framework of \cite{BPT20}. We will define a potential function, and show that it increases as $\kappa$ is increased, and edges are deleted. While their framework applies to simple graphs, we verify that it is possible to extend it to multigraphs as well. Towards this, we give a few definitions.

\begin{definition2}\label{def:mathcalM}
	Let $G$ be a multigraph, and let $\kappa$ be the capacity function on the edges of the graph (if the adversary deletes an edge $e$, then let $\kappa(e)$ be the capacity of the edge right before it is deleted). Let $\mu$ be the input to \textsc{Dec-Matching}(), when it is run on $G,\eps$. Let $\mathcal{M}$ be the set of all integral matchings of $G$ of size at least $(1-3\eps)\cdot \mu$. Define the cost of an edge $e$ to be $c(e)=\log\paren{n\cdot \kappa(e)}$. For any integral matching $M$, define $c(M)=\sum_{e\in M}c(e)$. Define $\Pi(G,\kappa)=\min_{M\in \mathcal{M}}c(M)$. If $\mathcal{M}=\emptyset$, then $\Pi(G,\kappa)=\infty$. 
\end{definition2}

\begin{observation2}
	Initially, $\kappa(e)=\nicefrac{1}{\alpha_{\eps}^{\lceil \log_{\alpha_{\eps}}n \rceil}}$, so $\Pi(G,\kappa)=0$. Note that capacities $\kappa$ only change (increase) in \Cref{line:capchange}. Consequently, the capacity function $\kappa$ is monotonically increasing. Moreover, edges of $G$ are only deleted. Thus, $\Pi(G,\kappa)$ is monotonically increasing as well. 
\end{observation2}

\begin{observation2}
	When $\Pi(G,\kappa)=\infty$, then \Cref{line:funcinfty} of \Cref{alg:decmatching}
	causes the algorithm to terminate.
\end{observation2}
\begin{proof}
	Suppose \Cref{line:funcinfty} doesn't cause \Cref{alg:proofthm1} to terminate, then, $\mu'\geqslant (1-3\eps)\cdot \mu$. Thus, $\mu(G)\geqslant \mu'\geqslant (1-3\eps)\cdot \mu$, and $\mathcal{M}\neq \emptyset$. 
\end{proof}

\begin{lemma2}
	In \textsc{Dec-Matching}(), we say that we begin a new phase when $\eps\cdot \mu$ value of the fractional matching $\vec{x}$ has been deleted (that is, the value of \textsc{Counterx} has increased to $\eps\cdot \mu$). Suppose we are in a phase when the algorithm does not terminate. Let $\kappa$ be the final capacities right before we process deletions. Then, $\Pi(G,\kappa)=O(\mu\log n)$.
\end{lemma2}
\begin{proof}
	First observe that for any edge $e$, $\kappa(e)\leqslant 1$, so $c(e)=O(\log n)$. This is implied by the fact that $E^*$ only consists of edges that have $\kappa(e)<1$ and $\kappa(e)$ is a power of $\alpha_{\eps}$ (see \Cref{lem:MorE}\ref{item:bi}). This is because, we start with $\kappa(e)$ to be a power of $\alpha_{\eps}$ initially (see \Cref{line:initcap}), and each time we increase $\kappa(e)$, we multiply it by $\alpha_{\eps}$. Moreover, for any matching $M$, we know that $\card{M}\leqslant \mu\cdot (1+\eps)$, since, $\mu\geqslant \mu(G)\cdot (1-\eps)$. Thus, $c(M)\leqslant (1+\eps)\cdot\mu\cdot \log n$. This implies that $\Pi(G,\kappa)=O(\mu\log n)$. 
\end{proof}

\begin{definition2}
	Let $E_0$ be the initial set of edges and let $\kappa(E_0)=\sum_{e\in E_0}\kappa(e)$.
\end{definition2}

\begin{lemma2}\label{lem:congestbal}
	Suppose a call to \textsc{M-or-E*}($G,\mu,\kappa,\eps$) returns the set $E^*$ instead of a matching. Let $\kappa'$ denote the new edge capacities after increasing capacities along $E^*$. Then,
	\begin{enumerate}[label=(\alph*)]
		\item\label{item:congesta} $\kappa'(E_0)\leqslant \kappa(E_0)+\alpha_{\eps} \cdot \mu\cdot \log n$, and
		\item\label{item:congestb} $\Pi(G,\kappa')\geqslant \Pi(G,\kappa)+\nicefrac{\mu}{\eps}$.
	\end{enumerate}
\end{lemma2}
\begin{proof}
	We begin by recalling \Cref{lem:MorE}\ref{item:edgeE} that $\kappa(E^*)\leqslant \mu\log n$. Thus, $\kappa'(E^*)=O(\alpha_{\eps}\cdot\mu\cdot \log n)$, since it is obtained by scaling up $\kappa(E^*)$ by a $\alpha_{\eps}$. Thus, we have,
	\begin{align*}
	\kappa'(E_0)&=\kappa(E_0\setminus E^*)+\kappa'(E^*)
	\end{align*}
	This implies that:
	\begin{align*}
	\kappa'(E_0)-\kappa(E_0\setminus E^*)&= \kappa'(E^*) = \alpha_{\eps}\cdot \mu\cdot \log n
	\end{align*}
	The LHS is lower bounded by $\kappa'(E_0)-\kappa(E_0)$, since $\kappa(E_0)\geqslant \kappa(E_0\setminus E^*)$. This proves the first part of our claim. To prove the second part, first notice that $E^*$ contains at least $\eps\cdot\mu$ edges of every matching of size at least $(1-3\eps)\cdot \mu$ (implied by \Cref{lem:MorE}\ref{item:edgeE}). Let $M,M'\in \mathcal{M}$ be the matchings that minimize $\Pi(G,\kappa)$ and $\Pi(G,\kappa')$ respectively, and let $c$ and $c'$ be the cost functions associated with $\kappa$ and $\kappa'$ respectively. Observe that $\card{M}\geqslant (1-3\eps)\cdot \mu$ and $\card{M'}\geqslant (1-3\eps)\cdot \mu$ (by definition of $\mathcal{M}$, see \Cref{def:mathcalM}). Then, we have the following.
	\begin{align*}
	\Pi(G,\kappa')&=c'(M')=\sum_{e\in M'} \log\paren{n\cdot \kappa'(e)}=\sum_{e\in M'\setminus E^*}\log \paren{n\cdot \kappa(e)}+ \sum_{e\in M'\cap E^*} \log\paren{n\cdot \kappa(e)\cdot \alpha_{\eps}}\\
	&=\sum_{e\in M'}\log \paren{n\cdot \kappa(e)}+\card{M'\cap E^*}\cdot \log \alpha_{\eps}	\\
	&= c(M')+ \nicefrac{\mu}{\eps}\geqslant \Pi(G,\kappa)+ \nicefrac{\mu}{\eps}\\
	&\text{(Since }\card{M'\cap E^*}\geqslant \eps\cdot \mu\text{ and }\alpha_{\eps}=2^{\nicefrac{32}{\eps^2}}) \\
	\end{align*}
	This proves our claim.
\end{proof}

\begin{observation2}\label{obs:calls1}
	The total number of calls to \textsc{M-or-E*}() (till \textsc{Dec-Matching}() terminates) that return $E^*$ are upper bounded by $O(\eps\cdot \log n)$. This is because the potential function $\Pi(G,\kappa)$ is upper bounded by $O(\mu\log n)$, and each call to \textsc{M-or-E*}() that returns $E^*$, increments $\Pi(G,\kappa)$ by $\nicefrac{\mu}{\eps}$ (see \Cref{lem:congestbal}\ref{item:congestb}). Thus, we can deduce that $\kappa(E_0)\leqslant \alpha_{\eps} \cdot \mu \cdot \eps \cdot \log n$. This is because each call to \textsc{M-or-E*}() that returns $E^*$ increases $\kappa(E_0)$ by at most $\alpha_{\eps}\cdot \mu\cdot \log n$ (see \Cref{lem:congestbal}\ref{item:congesta}), and there are at most $\eps\cdot \log n$ such calls. 
\end{observation2}

We now want to upper bound the number of calls made to \textsc{M-or-E*}() that return a matching. This is upper bounded by the number of phases. 
\begin{lemma2}\label{lem:phases}
	The total number of phases is at most $O(\alpha_{\eps}^3\cdot \log n)$. 
\end{lemma2}
\begin{proof}
	We define $\Phi_{\textsf{del}}\coloneqq\sum_{e\in E_{\textsf{del}}}\kappa(e)$ to be the capacity of deleted edges. Observe that $\Phi_{\textsf{del}}\leqslant \kappa(E_0)\leqslant \alpha_{\eps}\cdot \eps\cdot \mu\cdot \log n$. This is implied by the definition of $\kappa(E_0)$ and Observation \ref{obs:calls1}. Moreover, each time we start a new phase, at least $\eps\cdot \mu$ value of the fractional matching has been deleted. That is $\sum_{e\in E}x(e)$ has dropped by $\eps\cdot \mu$. From Property \ref{prop:capexceed}, one can conclude that for any edge $e$, $x(e)\leqslant \alpha_{\eps}^2\cdot \kappa(e)$. This implies that a phase contributes at least $\nicefrac{\eps\cdot \mu}{\alpha^2_{\eps}}$ to $\Phi_{\textsf{del}}$. Using the fact that $\Phi_{\textsf{del}}$ is upper bounded by $\alpha_{\eps}\cdot \eps\cdot \mu\cdot \log n$, we can conclude that the total number of phases is at most $\alpha_{\eps}^3\cdot \log n$. 
\end{proof}

\begin{proof}[Proof of \Cref{lem:runtime}]
	From \Cref{lem:phases} and Observation \ref{obs:calls1}, we conclude that the total number of calls to \textsc{M-or-E*}() are upper bounded by $O(\alpha_{\eps}^3\cdot \log n)$. Finally, the number of calls to \textsc{Static-Match}() due to \Cref{line:recomp} are upper bounded by the number of phases, which are at most $O(\alpha_{\eps}^3\cdot \log n)$ by \Cref{lem:phases}.
\end{proof}
\section{Ingredients for Algorithm \textsc{M-or-E*}()}

Recall that we use $\mu(G,\kappa)$ to denote the value of the maximum fractional matching of $G$ obeying capacity function $\kappa$ and the odd set constraints. As in the congestion balancing setup of \cite{BPT20}, we want to check if $\mu(G,\kappa)\geqslant (1-\eps)\cdot \mu(G)$. However, unlike in bipartite graphs, where we can use flows to find fractional matching, there is no simple way to check if $\mu(G,\kappa)\geqslant (1-\eps)\cdot \mu(G)$ in general graphs. Our first structural result circumvents this issue. Let $G_s$ be the graph obtained by sampling every edge $e$ with probability $p(e)=\min\set{1,\kappa(e)\cdot \rho_{\eps}}$. We show that $\mu(G_s)\geqslant \mu(G,\kappa)-\eps\cdot \mu(G)$. Thus, we can run \textsc{Static-Match}($G_s,\eps$) to estimate $\mu(G,\kappa)$. \\ \\
At a high level, \textsc{M-or-E*}() proceeds in three phases. In Phase 1, it creates $G_s$ and computes $\mu(G_s)$. If this matching is large, it proceeds to Phase 2, where it finds a fractional matching $\vec{x}$ such that $\sum_{e\in E}x(e)\geqslant (1-\eps)\cdot \mu(G)$. On the other hand, if $\mu(G_s)$ is small, then it proceeds to Phase 3, where it finds the set of edges $E^*$ along which it increases capacity. In the subsequent sections, we will state the main structural properties we use in each of the phases. Finally, in \Cref{sec:lemma7}, we put together these ingredients to give \textsc{M-or-E*}(), and prove \Cref{lem:MorE}.
\subsection{Phase 1 of \textsc{M-or-E*()}}

Before we formally state the main guarantees of Phase 1, we will state some standard results in matching theory, that we will use in our main result for Phase 1. 

\subsubsection{Some Standard Ingredients For Phase 1}

The first ingredient we use is the Tutte-Berge formula. 

\begin{definition2}
	Let $G$ be any graph (possibly containing multiedges), and let $U\subseteq V$, then \text{odd}$_{G}(V\setminus U)$ refers to the number of odd components in $G[V\setminus U]$. 
\end{definition2}

\begin{lemma2}[Tutte-Berge Formula]\cite{Schrijver2003CombinatorialOP}\label{lem:tuttewitn}
	The size of a maximum matching in a graph $G=(V,E)$ is equal to $\frac{1}{2}\min\limits_{U\subseteq V}\paren{\card{U}+\card{V}-\text{odd}_{G}(V\setminus U)}$.
\end{lemma2}

Additionally, we will use some properties of the matching polytope. 

\begin{lemma2}\cite{Schrijver2003CombinatorialOP}\label{lem:folklore}
	Let $G$ be any graph, and let $\vec{x}$ be a fractional matching that in addition to the fractional matching constraints, also satisfies the following for all odd-sized $U\subseteq V$: $\sum\limits_{e\in G[U]} x(e)\leqslant \frac{\card{U}-1}{2}$. Then, there is an integral matching $M\subseteq \text{supp}(x)$ with $\card{M}=\sum_{e\in E}x(e)$.
\end{lemma2}

%\begin{definition}
%	Recall that for an uncapacitated graph $G$, we let $\mu(G)$ denote the size of the maximum integral matching (or equivalently, the size of the maximum fractional matching satisfying odd set constraints). Similarly, for a graph $G$ with capacity function $\kappa$ on the edges, we let $\mu(G,\kappa)$ be the size of the maximum fractional matching obeying $\kappa$, and all odd set constraints. 
%\end{definition}

\begin{definition2}
	Let $G$ be any multigraph, and let $S,T\subseteq V$, then $\delta_{G}(S,T)$ is defined as the set of edges that have one endpoint in $S$ and the other in $T$. Additionally, for $S\subseteq V$, we define $\delta_G(S)$ to be the set of edges that have one end point in $S$, and the other in $V\setminus S$. 
\end{definition2}

\begin{figure}
	\centering
	\includegraphics[scale=0.25]{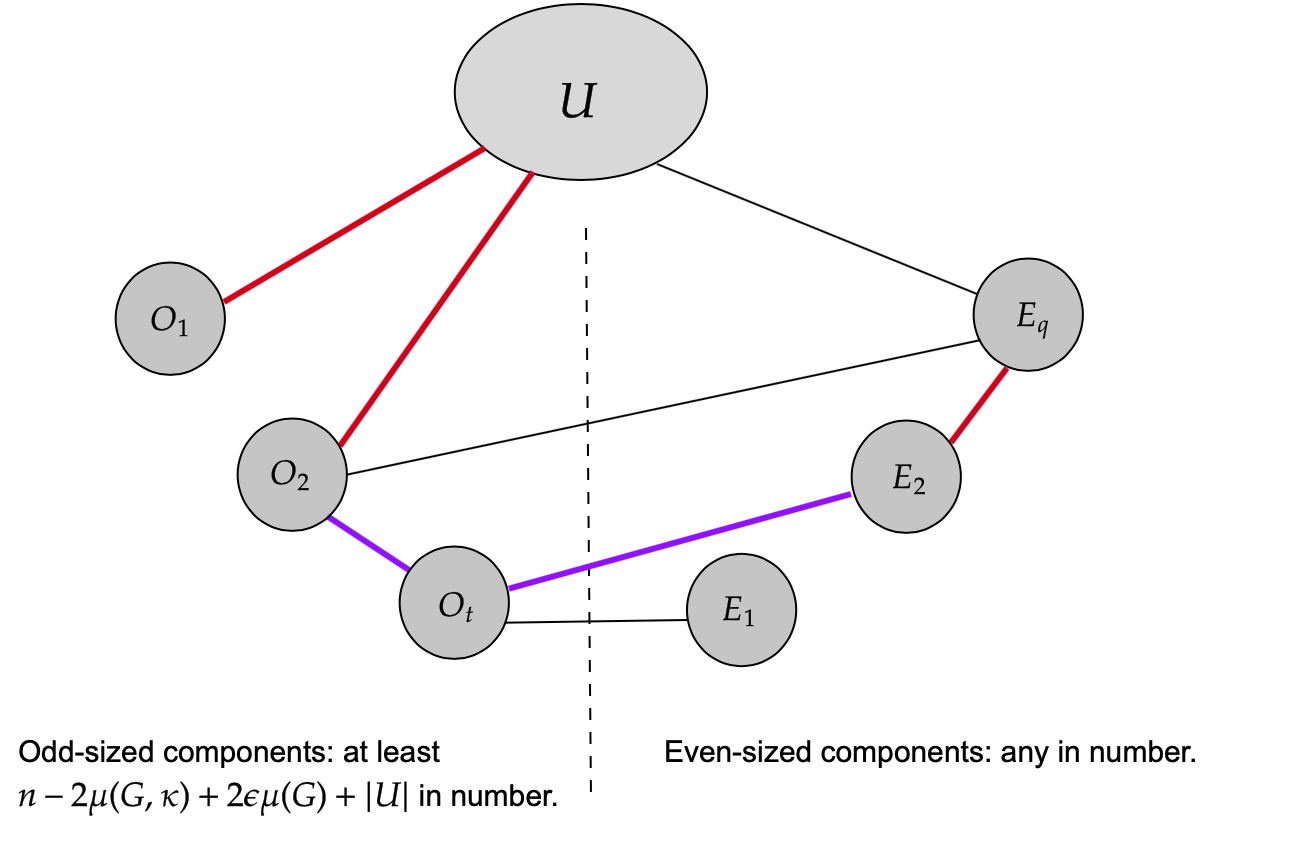}
	\caption{The figure shows the graph $G$, and a partition $\mathcal{P}=\set{U,E_1,\cdots, E_{q},O_1,\cdots, O_{t}}$ satisfying \ref{item:part1} and \ref{item:part2}. The thick edges (in red and purple), are the edges in supp($\vec{x}$), where $\vec{x}$ is the fractional matching realizing $\mu(G,\kappa)$. The purple edges (edges between the odd components, or between odd and even components) correspond to $E_{\textsf{miss}}^{\mathcal{P}}$, and $\sum_{e\in E _{\textsf{miss}}^{\mathcal{P}}}x(e)\geqslant 2\cdot \eps\cdot \mu(G)$. }
	\label{fig:tutte}
\end{figure}

\subsubsection{Main Lemma for Phase 1}

As mentioned earlier, in Phase 1 of \textsc{M-or-E*}(), we first create a sampled graph $G_s$. In the following lemma, we show that $\mu(G_s)$ is a good estimate for $\mu(G,\kappa)$ with high probability.

\begin{lemma2}\label{lem:matchingGs}
	Let $G$ be a multigraph with $\mu(G)\geqslant \nicefrac{\eps\cdot n}{16}$ where $\eps\in (0,\nicefrac{1}{2})$. Let $\kappa$ be a capacity function on the edges of $G$, and let $G_{s}$ be obtained by sampling every edge $e\in G$ with probability $p(e)=\min\set{\kappa(e)\cdot \rho_{\eps},1}$. Let $\mu(G,\kappa)$ be the value of the maximum fractional matching of $G$ obeying the capacities $\kappa$, and the odd set constraints. Then, with high probability, $\mu(G_s)\geqslant \mu(G,\kappa)-\eps\cdot\mu(G)$.
\end{lemma2}
\begin{proof} 
	We want to show that with probability at least $1-\nicefrac{1}{n^2}$, $\mu(G_s)\geqslant \mu(G,\kappa)-\eps\cdot \mu(G)$. In order to do this, by \Cref{lem:folklore}, it is sufficient to show that with probability at least $1-\nicefrac{1}{n^2}$, $\frac{1}{2}\min\limits_{U\subseteq V}\paren{\card{U}+\card{V}-\text{odd}_{G_s}(V-U)}\geqslant \mu(G,\kappa)-\eps\cdot \mu(G)$. \\ \\
	Towards this, we consider a fixed partition $\mathcal{P}$ of $V$ into sets $U,O_1,\cdots,O_{t},E_1,\cdots,E_{q}$ with the following properties (see \Cref{fig:tutte}).
	\begin{enumerate}[label=(\alph*)]
		\item \label{item:part1} We have, $q\geqslant 0$ and $t>n-2\cdot \mu(G,\kappa)+2\eps\cdot \mu(G)+\card{U}$. 
		\item \label{item:part2} Sets $O_i$ for $i\in \bracket{t}$ are odd-sized sets and sets $E_l$ for $l\in [q]$ are even-sized sets. 
	\end{enumerate}
Note that if $\mu(G_s)< \mu(G,\kappa)-\eps\cdot \mu(G)$, then there is a partition $\mathcal{P}=\set{U,O_1,\cdots,O_t,E_1,\cdots,E_{q}}$ of $G_s$, satisfying \ref{item:part1} and \ref{item:part2} such that $V\setminus U$ is the union of components $O_1,\cdots, O_t,E_1,\cdots, E_{q}$. Note that if $V\setminus U$ is the union of disjoint components $O_1,\cdots, O_{t},E_1,\cdots,E_{q}$ then $\delta_{G_s}(O_i,O_l)=\emptyset$ for all $i\neq l$ and $\delta_{G_s}(O_i,E_l)=\emptyset$ for all $i\neq l$ (this is evident from \Cref{lem:tuttewitn}). Thus, to upper bound the probability that $\mu(G_s)<\mu(G,\kappa)-\eps\cdot \mu(G)$, it is sufficient to upper bound the probability that for all partitions $\mathcal{P}=\set{U,O_1,\cdots,O_{t},E_1,\cdots,E_{q}}$ satisfying \ref{item:part1} and \ref{item:part2}, \emph{none} of the edges $E^{\mathcal{P}}_{\textsf{miss}}=\set{e\mid e\in \delta_{G}(O_i,O_l)\text{ for }i\neq l}\cup\set{e\mid e\in \delta_{G}(O_i,E_l)\text{ for }i\in [t],l\in[q]}$ are sampled in $G_s$. In order to bound this probability, we make the following claim. 
\begin{claim2} \label{claim:emiss}
For a partition $\mathcal{P}$ satisfying \ref{item:part1} and \ref{item:part2},  $\kappa(E^{\mathcal{P}}_\textsf{miss})\geqslant 2\cdot \eps\cdot \mu(G)$. 
\end{claim2}
\begin{proof}
Let $\vec{x}$ be a fractional matching obeying odd set constraints and capacity function $\kappa$ such that $\sum_{e\in E}x(e)=\mu(G,\kappa)$. In order to prove this claim, we show that if $\kappa(E_{\textsf{miss}}^{\mathcal{P}})<2\cdot \eps\cdot \mu(G)$, then, $x(E_{\textsf{miss}}^{\mathcal{P}})>\kappa(E_{\textsf{miss}}^{\mathcal{P}})$, which will contradict the fact that $\vec{x}$ is a fractional matching obeying $\kappa$.\\ \\
With this proof strategy in mind, for contradiction assume that $\kappa(E^{\mathcal{P}}_{\textsf{miss}})< 2\cdot \eps\cdot \mu(G)\leqslant n-2\cdot \mu(G,\kappa)+2\cdot \eps\cdot \mu(G)$. The last inequality follows from the fact that $\mu(G,\kappa)$ corresponds to the value of the maximum fractional matching, so, $\mu(G,\kappa)\leqslant \frac{n}{2}$. Since $\vec{x}$ obeys odd set constraints, $\sum_{l\leqslant t}x(\delta_G(O_l))\geqslant t$ (from \Cref{lem:folklore}). Note that $\sum_{l\leqslant t}x(\delta_G(O_{l},U))\leqslant \card{U}$, otherwise for some $v\in U$, $\sum_{e\ni v}x(e)>1$, violating the fact that $\vec{x}$ is a fractional matching. Next, we observe that $\sum_{l\leqslant t}x(\delta_G(O_l))=x(E^{\mathcal{P}}_{\textsf{miss}})+\sum_{l\leqslant t}x(\delta_G(O_l,U))$. This follows from the fact that all edges emanating out of $O_i$ in $G$, are either incident on other $O_j$, or $E_k$ or $U$. We have the following set of inequalities. 
\begin{align*}
x(E_{\textsf{miss}}^{\mathcal{P}})&=\sum_{l\leqslant t} x(\delta_{G}(O_l))-\sum_{l\leqslant t}x(\delta_{G}(O_l,U))\\
&\geqslant t-\card{U}\\
&> n-2\cdot \mu(G,\kappa)+2\cdot \eps\cdot \mu(G)+\card{U}-\card{U}
\end{align*}
Thus, that $x(E^{\mathcal{P}}_{\textsf{miss}})\geqslant n-2\mu(G,\kappa)+2\eps\mu(G)>\kappa(E^{\mathcal{P}}_{\textsf{miss}})$. This is a contradiction. This concludes our proof, and we know that $\kappa(E^{\mathcal{P}}_{\textsf{miss}})\leqslant 2\cdot \eps\cdot \mu(G)$.
\end{proof}
Additionally, we have the following claim, which allows us to only focus on $\mathcal{P}$ for which all $e\in E^{\mathcal{P}}_{\textsf{miss}}$ have $\kappa(e)<\nicefrac{1}{\rho_{\eps}}$.

\begin{observation2}\label{claim:emiss2}
	Suppose $\kappa(e)\geqslant \nicefrac{1}{\rho_{\eps}}$ for any $e\in E^{\mathcal{P}}_{\textsf{miss}}$, then,
	$$\prob{\text{None of the edges in }E^{\mathcal{P}}_{\textsf{miss}}\text{ are sampled in }G_s}=0.$$
\end{observation2}
The above observation follows from the fact that an edge $e$ is sampled with probability $\min\set{1,\kappa(e)\cdot \rho_{\eps}}$. From the above claim, we can deduce that if for any partition $\mathcal{P}$, if $E_{\textsf{miss}}^{\mathcal{P}}$ has an edge $e$ with $\kappa(e)\geqslant \nicefrac{1}{\rho_{\eps}}$, then the contribution of $E_{\textsf{miss}}^{\mathcal{P}}$ to our probability bound will be $0$. Hence, it is sufficient to focus on $E_{\textsf{miss}}^{\mathcal{P}}$ where all edges $e$ have $\kappa(e)< \nicefrac{1}{\rho_{\eps}}$. We now bound the following probability.
	\begin{align*}
	\Prob\paren{\text{None of the edges in }{E^{\mathcal{P}}_{\textsf{miss}}\text{ are sampled in }G_s}}&\leqslant \prod_{e\in E^{\mathcal{P}}_{\textsf{miss}}}(1-p(e))\\
	&\leqslant \exp\paren{-\sum_{e\in E^{\mathcal{P}}_{\textsf{miss}}}p(e)}\\
	&=\exp\paren{-\sum_{e\in E^{\mathcal{P}}_{\textsf{miss}}}\kappa(e)\cdot \rho_{\eps}}\\
	&\text{(Since }p(e)=\kappa(e)\cdot \rho_{\eps}\text{ (from Claim \ref{claim:emiss2} and discussion above))}\\
	&=\exp\paren{-\eps\cdot 2^{\nicefrac{16}{\eps^2}}\cdot \mu(G)\cdot \log n}\\
	&\text{(From Claim \ref{claim:emiss} and the fact that }\rho_{\eps}=2^{\nicefrac{16}{\eps^2}}\cdot \log n)\\
	&\leqslant \exp\paren{- 2^{\nicefrac{15}{\eps^2}}\cdot \mu(G)\cdot \log n}.
	\end{align*}
	Note that the total number of partitions of graph $G$ are upper bounded by $2^{n\cdot \log n}$. Hence, this also upper bounds the total number of partitions of $G$ satisfying \ref{item:part1} and \ref{item:part2}. Since $n\leqslant \nicefrac{16\cdot\mu(G)}{\eps}$ and $\eps<\nicefrac{1}{2}$, the bound on the number of partitions is at most $2^{\nicefrac{16\mu}{\eps}\cdot \log n}$ (by assumption), taking a union bound over all the partitions, we know that the with probability at least $1-\exp\paren{-\nicefrac{\mu(G)\cdot \log n}{\eps}}$, in $G_s$, we have no partition $\mathcal{P}=\set{U,O_1,\cdots, O_t,E_{1},\cdots, E_{q}}$ satisfying \ref{item:part1} and \ref{item:part2}. Thus, by \Cref{lem:tuttewitn}, we have that with high probability, $\mu(G_s)\geqslant \mu(G,\kappa)-\eps\cdot \mu(G)$. 
\end{proof}

\subsection{Phase 2 of \textsc{M-or-E*}()}

The algorithm proceeds to Phase 2 only if the integral matching $M_s$ found in $G_s$ is close to $\mu(G,\kappa)$. Recall that our goal is to compute a fractional matching so that we can apply the congestion balancing framework. However, as mentioned before there is no straightforward way of computing a maximum fractional matching in a general graph obeying capacity $\kappa$. To overcome this, we give a candidate fractional matching which is easy to compute, and is sufficient for our purposes (that is, it avoids the integrality gap). At a high level, this is what Phase 2 does, and in this section we describe our candidate fractional matching and show that it is close in value to $\mu(G_s)$ with high probability, and therefore it is close to $\mu(G,\kappa)$ as well (by \Cref{lem:matchingGs}). 
\subsubsection{Preliminaries for Phase 2}

Phase 2 starts by computing $M_s = \textsc{Static-Match}(G_s, \eps)$ and then uses $M_s$ to compute the desired fractional matching. We will split $M_s$ into low capacity edges and high capacity edges, and as a result split $V$ into vertices matched using high capacity edges, and low capacity edges. We begin by giving a formal definition of low capacity edges. 

\begin{definition2}\label{def:lowedges}
Let $G$ be any multigraph, and let $\kappa$ be a capacity function on the edges of $G$. Let $\eps\in (0,\nicefrac{1}{2})$. Define $E_{L}(G,\kappa)=\set{e\in E\mid e\in D(u,v)\text{ and }\kappa(D(u,v))\leqslant \nicefrac{1}{\alpha^2_{\eps}}}$. Intuitively, $E_{L}(G,\kappa)$ is the set of low total capacity edges. 
\end{definition2}

As mentioned in the high-level overview, in order to prove our probabilistic claims, we will give some slack to the capacities. This motivates our next definition.

\begin{definition2}\label{def:kappaplus}
	Let $G$ be a multi-graph, and let $\kappa$ be a capacity function on the edges of $G$. Let $\eps\in (0,\nicefrac{1}{2})$. We define the capacity function $\kappa^+$ as follows.
	\begin{enumerate}[label=(\alph*)]
		\item For all $e\in E_{L}(G,\kappa)$, $\kappa^{+}(e)=\kappa(e)\cdot \alpha_{\eps}$.
		\item For all $e\in E\setminus E_{L}(G,\kappa)$, $\kappa^+(e)=\kappa(e)$. 
	\end{enumerate}
\end{definition2}

To make our analysis easier to follow, we need the following definition of a bipartite double cover of $G$.

\begin{definition2}
	Let $G$ be a multi-graph and let $\kappa$ be a capacity function on the edges of $G$. We define the bipartite double cover $\textsc{bc}(G)$ to be a bipartite graph with capacity function $\kappa_{\textsc{bc}}$ as follows.
	\begin{enumerate}[label=(\alph*)]
		\item For every vertex $v\in V(G)$, make two copies $v$ and $v'$ in $V(\textsc{bc}(G))$. 
		\item If $e$ is an edge between $u,v\in V(G)$, then for each such $e$ we add two edges $e'$ and $e''$, one between $u$ and $v'$ and the other between $v$ and $u'$. We let $\kappa_{\textsc{bc}}(e')=\kappa_{\textsc{bc}}(e'')=\kappa(e)$.
	\end{enumerate}
\end{definition2}
We have the following standard claim relating $\mu(G)$ and $\mu(\textsc{bc}(G))$. \begin{claim2}\label{claim:matchbc}
	For any multi-graph $G$, $\mu(\textsc{bc}(G))\geqslant 2\cdot \mu(G)$.
\end{claim2}

Next, we state the following lemma, which follows from standard techniques, and we give a formal proof of it in \Cref{sec:missingproof}. The lemma essentially states that a fractional matching which has low flow on all edges has a very small integrality gap. 

\begin{restatable}{lemma2}{lemnono}\label{lem:folklore2}
	Let $G$ be a multigraph, and let $\eps\in (0,1)$. Let $\kappa$ be a capacity function on the edges of $G$, with $\kappa(D(e))\leqslant \nicefrac{1}{\alpha_{\eps}}$ for all $e\in E(G)$. Then, $\mu(\textsc{bc}(G),\kappa_{\textsc{bc}})\leqslant 2\cdot\paren{1+\eps}\cdot \mu(G,\kappa)$, where $\mu(G,\kappa)$ is the maximum fractional matching of $G$ obeying $\kappa$ and the odd set constraints, and $\mu(\textsc{bc}(G),\kappa_{\textsc{bc}})$ is the maximum fractional matching of $\textsc{bc}(G)$ obeying $\kappa_{\textsc{bc}}$. 
\end{restatable}

Additionally, we will need the following version of Hall's theorem, which follows as a corollary of \Cref{lem:tuttewitn}.

\begin{proposition}[Extended Hall's Theorem]\label{prop:halls}
	Let $G=(L\cup R, E)$ be a bipartite graph with $n=\card{L}=\card{R}$, then $\mu(G)=n-\max_{S\subseteq L}\paren{\card{S}-\card{N_G(S)}}$, where $N_G(S)$ refers to the neighbourhood of $S$ in $G$. 
\end{proposition}

In order to prove \Cref{lem:sampling2}, we will use the following version of Chernoff bound.

\begin{lemma2}[Chernoff Bound]\label{lem:chernoff}
	Let $X_1,\cdots, X_k$ be negtively correlated random variables, and let $X$ denote their sum, and let $\mu=\expect{X}$. Suppose $\mu_{\text{min}}\leqslant \mu\leqslant \mu_{\text{max}}$, then for all $\delta>0$,
	$\Prob\paren{X\geqslant (1+\delta)\mu_{\text{max}}}\leqslant \paren{\frac{e^{\delta}}{(1+\delta)^{\delta}}}^{\mu_{\text{max}}}$.
\end{lemma2}

Additionally, we state the following observation, we refer the readers to \cite{BPT20} for a proof of this observation. The proof follows from a standard application of the max-flow min-cut theorem, and we refer the reader to \cite{BPT20} for a proof sketch.

\begin{observation2}\label{obs:maxflow}
	Let $G$ be any bipartite multigraph, with vertex bipartitions $S$ and $T$. Let $\kappa$ be the capacity function on the edges of the graph. Then, for any $C\subseteq S$, and $D\subseteq T$, we have, $\card{S}-\card{C}+\card{D}+\kappa(C,T\setminus D)\geqslant\mu(G,\kappa)$. Moreover, there are sets $C\subseteq S$ and $D\subseteq T$ such that equality holds. 
\end{observation2}

\subsubsection{Main Result for Phase 2}

We briefly give some intuition about the statement of our claim. Recall in the high level review, we mentioned that $M$, the integral matching of $G_s$ can split into two parts $M_H$, which is the high capacity part, and $M_L$, the low capacity part, and we defined $V_H=V(M_H)$ and $V_L=V(M_L)$. We said that congestion balancing allows us to give slack to capacities, and therefore, we can round up the capacities of $M_H$ to $1$. However, we still want to compute a fractional matching $G[V_L]$. In order to do this, we observe that if the fractional matching in $G[V_L]$ is only on low capacity edges, then we can use flow algorithms to compute such a matching. Therefore, our main structural result for Phase 2 states that if $\vec{y}$ is a fractional matching on the low capacity edges of $G[V_L]$, then with high probability $\sum_{e\in E}y(e)\geqslant \card{M_L}-\eps\cdot \mu(G)$. We now state this result formally.

\begin{figure}
	\centering
	\includegraphics[scale=0.25]{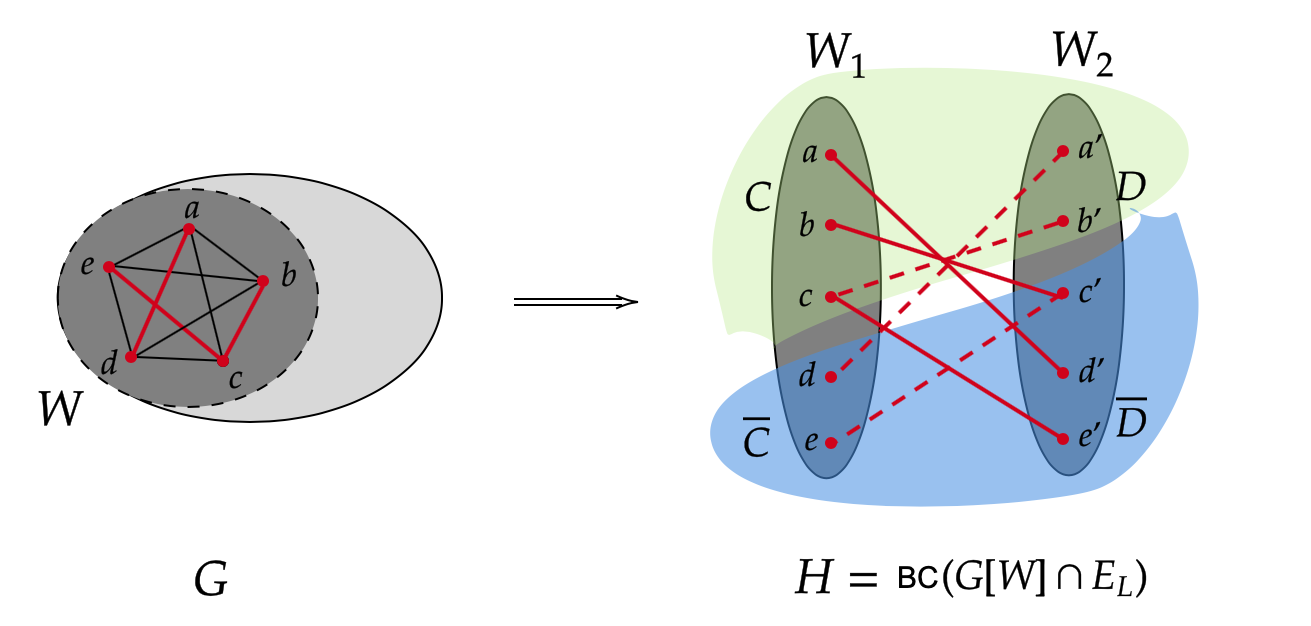}
	\caption{In the left panel, we consider the graph $G$, and $W=\set{a,b,c,d,e}$. The red edges correspond to the low capacity edges of $G[W]$, denoted $G[W]\cap E_{L}$. On the right panel, we have the bipartite graph $\textsc{bc}(G[W]\cap E_L)$, which has bipartitions $W_1$ and $W_2$. The solid red edges are the edges going between $C$ and $\bar{D}$, and these edges have capacity $\kappa(C,\bar{D})$. }
	\label{fig:sampling}
\end{figure}

\begin{lemma2}\label{lem:sampling2}
	Let $G$ be a multigraph and let $\eps\in (0,\nicefrac{1}{2})$ and suppose $\mu(G)\geqslant \nicefrac{\eps\cdot n}{16}$. Let $\kappa$ be a capacity function on the edges, and let $G_s$ be the graph obtained from $G$ by sampling each edge $e$ with probability $p(e)=\kappa(e)\cdot \rho_{\eps}$. Let $E_{L}\coloneqq E_{L}(G,\kappa)$. Then, for all $W\subseteq V$, we have, with high probability,
	\begin{align*}
	\mu(\textsc{bc}\paren{G_{s}[W]\cap E_{L}})\leqslant \mu(\textsc{bc}\paren{G[W]\cap E_{L}},\kappa^{+}_{\textsc{bc}})+\eps\mu(G).
	\end{align*}
\end{lemma2}

\begin{remark2}
	Note that by definition of $E_{L}$, all edges $e\in G[W]\cap E_{L}$ have $\kappa(e)\leqslant \nicefrac{1}{\alpha_{\eps}^2}$. Thus, $\kappa^+_{\textsc{bc}}(e)=\kappa_{\textsc{bc}}(e)\cdot \alpha_{\eps}$ for all $e\in \textsc{bc}(G[W]\cap E_{L})$.
\end{remark2}

Before we prove it, we have the following statement as a corollary of \Cref{lem:sampling2}.

\begin{corollary2}\label{cor:sampling}
Let $G$ be a multi-graph, and let $\eps\in (0,1)$. Let $\kappa$ be a capacity function on the edges of $G$ with $\kappa(e)\leqslant \nicefrac{1}{\alpha_{\eps}}$ for all $e\in E(G)$. Then, with high probability, \textbf{for all} $W\subseteq V$, $\mu(G_{s}[W]\cap E_{L})\leqslant \mu(G[W]\cap E_{L},\kappa^+)+2\eps\mu(G)$.
\end{corollary2}
\begin{proof}
	The inequality in the statement follows due to the following line of reasoning.
	\begin{equation*}
	\begin{split}
	2\cdot \mu(G_{s}[W]\cap E_{L})&\leqslant \mu(\textsc{bc}\paren{G_{s}[W]\cap E_{L}})\\
	&\text{(since $2\cdot \mu(H)\leqslant \mu(\textsc{bc}(H))$, see Claim \ref{claim:matchbc})}\\
	&\leqslant \mu(\textsc{bc}(G[W]\cap E_{L}),\kappa^+_{\textsc{bc}})+\eps\mu(G)\\
	&\text{(by \Cref{lem:sampling2})}\\
	&\leqslant 2\cdot(1+\eps)\cdot \mu(G[W]\cap E_{L},\kappa^+)+\eps\mu(G)\\
	&\text{(by \Cref{lem:folklore2})}\\
	&\leqslant 2\cdot \mu(G[W]\cap E_{L},\kappa^+) +3\eps\mu(G).
	\end{split}
	\end{equation*}
	The last inequality follows from the fact that any fractional matching in $G$ obeying $\kappa^+$ and odd set constraints is upper bounded by $\mu(G)$ (by \Cref{lem:folklore}). Dividing by two on both sides, we have our claim. 
\end{proof}

\begin{proof}[Proof of \Cref{lem:sampling2}]
	Consider a fixed $W\subseteq V$. From now, we will use $H$ to denote $\textsc{bc}(G[W]\cap E_{L})$ and let $H_s$ denote $\textsc{bc}(G_s[W]\cap E_{L})$. For the bipartite graph $H$, we will use $W_1$ and $W_2$ to denote the two bipartitions of $H$ corresponding to $W$ (see \Cref{fig:sampling} for an illustration). We now want to prove that $\mu(H_s)\leqslant \mu(H,\kappa^+_{\textsc{bc}})+\eps\mu(G)$. To prove this, it is sufficient to show a set $C\subseteq W_1$ such that $\card{C}-\card{N_{H_s}(C)}\geqslant \card{W_1}-\mu(H,\kappa^+_{\textsc{bc}})-\eps\mu(G)$ with high probability. Then, by \Cref{prop:halls}, we have the following inequality.
	\begin{align*}
	\card{W_1}-\mu(H,\kappa^+_{\textsc{bc}})-\eps\mu(G) \leqslant \card{C}-\card{N_{H_s}(C)}\leqslant \card{W_1}-\mu(H_{s}).
	\end{align*}
	This would prove our claim. Towards this, we consider the set $C\subseteq W$ satisfying the following equation (applying Observation \ref{obs:maxflow} to $H$ with $\kappa_{\textsc{bc}}^+$ on edges), and show that this is the required set.
	\begin{align}\label{eqn:upper}
	\card{W_1}-\card{C}+\card{D}+\kappa_{\textsc{bc}}^+(C,W_2\setminus D)=\mu(H,\kappa^+_{\textsc{bc}}).
	\end{align}
	We want to show that $\card{C}-\card{N_{H_s}(C)}\geqslant\card{W_1}-\mu(H,\kappa^+_{\textsc{bc}})-\eps\mu(G)$ with high probability. Let $L$ be the set of vertices in $N_{H_s}(C)\cap W_2\setminus D$. We know that $\card{N_{H_s}(C)}\leqslant \card{D}+\card{L}$. \\ \\
	Let $E_{H}(C,W_2\setminus D)$ be the set of edges in $H$ between $C$ and $W_2\setminus D$. Let $X$ be the random variable that denotes the number of edges in $E_{H}(C,W_2\setminus D)$ that are sampled in $H_s$. Note that $X$ is not a sum of independent random variables. Recall that $H$ is a subgraph of $\textsc{bc}(G)$, and suppose $e\in G[W]\cap E_{L}$ is included in $G_s$, then $e'$ and $e''$ (recall $e'$ and $e''$ are copies of $e$ in \textsc{bc}($G$)) are both included in $H_s$ else both are excluded. Thus, the random variables associated with $e'$ and $e''$ are correlated with each other. We instead consider an arbitrary subset of $E^*_{H}(C,W_2\setminus D)$ of $E_{H}(C,W_2\setminus D)$ that satisfies the following properties.
	\begin{enumerate}[label=(\alph*)]
		\item For $e\in G[W]$, if $\set{e',e''}\subset E_{H}(C,W_2\setminus D)$, then exactly one of $e'$ or $e''$ is included in $E^*_{H}(C,W_2\setminus D)$.
		\item For $e\in G[W]$, if $\set{e',e''}\cap E_{H}(C,W_2\setminus D)=\set{e'}$, then only $e'$ is included in $E^*_{H}(C,W_2\setminus D)$.
		\item For $e\in G[W]$, if $\set{e',e''}\cap E_{H}(C,W_2\setminus D)=\set{e''}$, then only $e''$ is included in $E^*_{H}(C,W_2\setminus D)$.
	\end{enumerate}
	We consider the random variable $Y$ that denotes the number of edges in $E^*_{H}(C,W_2\setminus D)$ that are included in $H_s$. Observe that $Y$ is a sum of independent random variables satisfying the condition of \Cref{lem:chernoff}. Moreover, $X\leqslant 2Y$. Thus, it is sufficient to upper bound the value $Y$ can take with high probability. Note that for any $e\in E_{H}(C,W_2\setminus D)$, using the definition of $H$, $\kappa_{\textsc{bc}}(e)<\nicefrac{1}{\alpha^2_{\eps}}$. Since $\rho_{\eps}<\alpha_{\eps}$, $p(e)=\rho_{\eps}\cdot \kappa_{\textsc{bc}}(e)<1$. So, we have,
	\begin{align*}
	\expect{Y}&\leqslant \sum_{e\in E_{H}^*(C,W_2\setminus D)} p(e)\leqslant  \rho_{\eps}\cdot \kappa_{\textsc{bc}}(C,W_2\setminus D)\leqslant \frac{\kappa^+_{\textsc{bc}}(C,W_2\setminus D)}{2^{\nicefrac{16}{\eps^2}}}.
	\end{align*}
	The last inequality follows from the fact that in $H$, $\kappa^+_{\textsc{bc}}(e)=\kappa_{\textsc{bc}}(e)\cdot \alpha_{\eps}$ for all $e\in H$. By definition of $H$ and \Cref{def:kappaplus}, for all edges $e\in H$, the corresponding original edge in $G$ is in $E_{L}(G,\kappa)$. We want to bound the following probability.
	\begin{align*}
	\prob{Y\geqslant \frac{\kappa^+_{\textsc{bc}}(C,W_2\setminus D)}{2^{\nicefrac{16}{\eps^2}}}+4\cdot\eps\mu(G)}
	\end{align*}
	Applying \Cref{lem:chernoff} with $\delta=\frac{4\cdot\eps\mu(G)\cdot 2^{\nicefrac{16}{\eps^2}}}{\kappa^+_{\textsc{bc}}(C,W_2\setminus D)}$, we have,
	\begin{align*}
	\prob{Y\geqslant \frac{\kappa^+_{\textsc{bc}}(C,W_2\setminus D)}{2^{\nicefrac{16}{\eps^2}}}+4\cdot\eps\mu(G)}&=\exp\paren{\eps\mu(G)-\eps\mu(G)\log\paren{1+\frac{4\cdot\eps\mu(G)\cdot 2^{\nicefrac{16}{\eps^2}}}{\kappa^{+}_{\textsc{bc}}(C,W_2\setminus D)}}}\\
	&\text{(Since }\kappa^{+}_{\textsc{bc}}(C,W_2\setminus D)\leqslant 4\cdot \mu(G)\text{ as proved below.)} \\
	&\leqslant \exp\paren{\eps\mu(G)-\eps\mu(G)\log\paren{1+\eps\cdot 2^{\nicefrac{16}{\eps^2}}}}\\
	&\text{(From \Cref{eqn:upper}, we have $\kappa^+_{\textsc{bc}}(C,W_2\setminus D)\leqslant 4\cdot\mu(G)$)}\\
	&\leqslant\exp\paren{\eps\mu(G)-\eps\mu(G)\log\paren{1+ 2^{\nicefrac{16}{\eps^2}}}}\\
	&\text{(Using the fact that }2^{\nicefrac{1}{\eps^2}}\geqslant \nicefrac{1}{\eps})\\
	&=\exp\paren{\eps\mu(G)-\nicefrac{16\mu(G)}{\eps}}. \\
	&\text{(Using the fact that }2^{\nicefrac{16}{\eps^2}}\leqslant 2^{\nicefrac{16}{\eps^2}}+1)
	\end{align*}
	To see why $\kappa^+_{\textsc{bc}}(C,W_2\setminus D)\leqslant 4\cdot\mu(G)$, consider \Cref{eqn:upper},
	\begin{align*}
	\kappa^+_{\textsc{bc}}(C,W_2\setminus D)&=\mu(H,\kappa^+_{\textsc{bc}})-\card{W_1}+\card{C}-\card{D}\\
	&\leqslant 2\cdot (1+\eps)\cdot\mu(G[W]\cap E_{L}(G,\kappa),\kappa^+_{\textsc{bc}})-\card{W_1}+\card{W_1}\\
	&\text{(Using \Cref{lem:folklore2}, the fact that }H=\textsc{bc}(G[W]\cap E_{L}(G,\kappa)),\text{ and $\kappa^+$ satisfies the hypothesis.})\\
	&\leqslant 4\cdot \mu(G).
	\end{align*}
	Taking a union bound over all $W$, which are at most $2^{\nicefrac{16\cdot \mu(G)}{\eps}}$ many, we have our bound (this is because by statement of the lemma, $\mu(G)\geqslant \nicefrac{\eps\cdot n}{16}$). 
\end{proof}

\subsection{Phase 3: Finding set $E^*$}

The algorithm \textsc{M-or-E*}() proceeds to Phase 3 if the matching found in $G_s$ in Phase 1 is small, and hence $\mu(G,\kappa)$ is too small. In this case, we need to find a set $E^*$ satisfying the properties of \Cref{lem:MorE}. In particular, we need to find a set $E^*$ with $\kappa(E^*)=O(\mu(G)\log n)$ such that for every large matching $M$, there are a lot of edges going through $E^*$. In order to do this, we rely on the properties of the dual variables associated with the matching problem. The algorithm \textsc{Static-Match}(), luckily for us, solves both the primal as well as the dual solution. We first begin by stating the Dual Matching Program, and then we state the properties of dual variables guaranteed by \textsc{Static-Match}().

\paragraph*{Dual Matching Program} The dual of the maximum cardinality matching linear program is the following (see \cite{Schrijver2003CombinatorialOP}).
\begin{mini*}
	{}{\sum_{v\in V} y(u)+\sum\limits_{B\subseteq V_{\text{odd}}}\frac{\card{B}-1}{2}z(B)}{}{}
	\addConstraint{yz(e)}{\geqslant 1, }{}{\quad \quad e\in E(G)}
	\addConstraint{y(u)}{\geqslant 0}{}{\quad \quad \forall u\in U}
	\addConstraint{z(B)}{\geqslant 0}{}{\quad \quad \forall B\subseteq V_{\text{odd}}}
\end{mini*}
where,
\begin{enumerate}
	\item We define $yz((u,v))\coloneqq y(u)+y(v)+\sum\limits_{\substack{B\in V_{\text{odd}},\\ (u,v)\in G[B]}} z(B)$ and, 
	\item We define $f(y,z)\coloneqq \sum_{v\in V} y(u)+\sum_{B\subseteq V_{\text{odd}}}\frac{\card{B}-1}{2}z(B)$.
\end{enumerate}
We now describe some properties of the \textsc{Static-Match}(). We state them without proof for now, and postpone the full proof to the \Cref{app:statmatch}. 

\begin{restatable}{lemma2}{statmatch}\label{lem:propstatmatch}
	There is an $O(\nicefrac{m}{\eps})$ time algorithm \textsc{Static-Match}() that takes as input a simple graph $G$ with $m$ edges and a parameter $\eps>0$, and returns a matching $M$, and dual vectors $\vec{y}$ and $\vec{z}$ that have the following properties.
	\begin{enumerate}[label=(\alph*)]
		\item \label{item:a}It returns an integral matching $M$ such that $\card{M}\geqslant (1-\eps)\cdot \mu(G)$.
		\item \label{item:b} A set $\Omega$ of laminar odd-sized sets, such that $\set{B\mid z(B)>0}\subseteq \Omega$.
		\item \label{item:c} For all odd-sized $B$  with $\card{B}\geqslant \nicefrac{1}{\eps}+1$, $z(B)=0$. 
		\item \label{item:d} Each $y(v)$ is a multiple of $\eps$ and $z(B)$ is a multiple of $\eps$. 
		\item \label{item:e} For every edge $e\in E$, $yz(e)\geqslant 1-\eps$. We say that such an edge $e$ is \emph{approximately covered} by $\vec{y}$ and $\vec{z}$. 
		\item \label{item:folklore}The value of the dual objective, $f(y,z)$ is at most $(1+\eps)\cdot \mu(G)$. 
	\end{enumerate}
\end{restatable}

\subsubsection{Main Guarantees of Phase 3}

We now state the following helper lemma which will be instrumental in proving one of the two main properties of $E^*$, namely, that every large matching has a lot of edges passing through $E^*$.

\begin{lemma2}\label{lem:guaran2}
	Suppose $G$ is a graph, and let $H\subset G$ be a subgraph of $G$. Let $\vec{y},\vec{z}$ be the dual variables returned on execution of \textsc{Static-Match}() on $H$ and $\eps$. Let $E_{H}=\set{e\in G\mid yz(e)\geqslant 1-\eps}$. Let $M$ be any matching of $G$, then $E\setminus E_{H}$ contains at least $\card{M}-(1+\eps)^2\cdot \mu(H)$ edges of $M$.
\end{lemma2}
\begin{proof}
	Suppose we scale up the dual variables $\vec{y}$ and $\vec{z}$ by a factor of $1+\eps$. Then, $\vec{y}$ and $\vec{z}$ is a feasible solution for the dual matching program for the graph $E_{H}$. Thus, using weak duality, we have that $\mu(E_{H})\leqslant (1+\eps)^2\cdot\mu(H)$ (this inequality follows from \Cref{lem:propstatmatch}\ref{item:folklore}). Suppose $M$ is a matching of $G$, and suppose $E\setminus E_{H}$ contains fewer than $\card{M}-(1+\eps)^2\cdot \mu(H)$ edges of $M$, and therefore, fewer than $\card{M}-\mu(E_H)$ edges of $M$, then this implies that $\card{M}=\card{M\cap E_{H}}+\card{M\cap E\setminus E_H}<\mu(E_{H})+\card{M}-\mu(E_{H})=\card{M}$, which is a contradiction. Thus, $E\setminus E_H$ contains at least $\card{M}-\mu(E_H)\geqslant \card{M}-(1+\eps)^2\cdot \mu(H)$ edges of $M$. 
\end{proof}
We now state set $E^*$, and show that it has the property that $\kappa(E^*)=O(\mu(G)\log n)$.

\begin{lemma2}\label{lem:propE*}
	Let $G$ be a graph multi-graph such that $\mu(G)\geqslant \nicefrac{\eps\cdot n}{16}$, and let $\kappa$ be a capacity function on the edges of the graph. Suppose $G_s$ is the graph obtained by sampling every edge $e\in E$ with probability $p(e)=\kappa(e)\cdot \rho_{\eps}$. Let $\vec{y}, \vec{z}$ be the duals returned by \textsc{Static-Match}($G_s,\eps$). Let $E^*=\set{e\in E\mid yz(e)<1-\eps}$. Then, with high probability, $\kappa(E^*)=O(\mu(G)\log n)$. 
\end{lemma2}
\begin{proof}
	We consider the set $\mathcal{D}$ of assignments $\vec{y},\vec{z}$ to the vertices and odd components that satisfy the following properties. 
	\begin{enumerate}[label=(\alph*)]
		\item \label{item:one} For all $v\in V$, $y(v)$ is a multiple of $\eps$, and for all $B\subset V$, $z(B)$ is a multiple of $\eps$.
		\item Let $\Omega=\set{B\subset V\mid z(B)>0}$, then $\Omega$ is laminar. 
		\item\label{item:three} If $z(B)>0$ for some $B\subseteq V$, then $\card{B}\leqslant \nicefrac{1}{\eps}$. 
	\end{enumerate}
	Observe that,
	\begin{align*}
	\card{\mathcal{D}}&\leqslant \sum_{i=0}^{2n} {n^{\nicefrac{1}{\eps}}\choose i}\cdot \paren{\frac{1}{\eps}}^n\cdot \paren{\frac{1}{\eps}}^{2n}\\
	&\leqslant n\cdot {n^{\nicefrac{1}{\eps}}\choose 2n}\cdot \paren{\frac{1}{\eps}}^n\cdot \paren{\frac{1}{\eps}}^{2n}\\
	&\leqslant 2^{\paren{\nicefrac{4}{\eps}}\cdot n\log n}\\
	&\leqslant 2^{\paren{\nicefrac{16}{\eps^2}}\cdot \mu(G)\log n}\\
	&\text{(Since }\mu(G)\geqslant \nicefrac{\eps\cdot n}{16})
	\end{align*}
	This number follows from the following argument. Since $\Omega$ is laminar, it can contain at most $2n$ sets. Moreover, from \ref{item:three}, we deduce that these $2n$ sets are chosen from among $n^{\nicefrac{1}{\eps}}$ sets. Further, from \ref{item:one}, we deduce that every vertex can be assigned at most $\nicefrac{1}{\eps}$ values, and every $B\in \Omega$ can be assigned $\nicefrac{1}{\eps}$ values. Therefore, for a given choice of $\Omega$, there are at most $\paren{\nicefrac{1}{\eps}}^n\cdot \paren{\nicefrac{1}{2\eps}}^{2n}$ choices for $\vec{y}$ and $\vec{z}$. Moreover, the above number also upper bounds the number of possible duals that can be a returned by the algorithm \textsc{Static-Match}(), since the duals in $\mathcal{D}$ satisfy a subset of the properties satisfied by the duals returned by \textsc{Static-Match}(). \\ \\
	Now consider a fixed $\vec{y},\vec{z}$ that satisfies the above-mentioned properties, and let $E^*$ be the set of edges that are not approximately covered by $\vec{y},\vec{z}$. Note that for any $e\in E^*$, $\kappa(e)\leqslant \nicefrac{1}{\alpha_{\eps}}$. If not, then, $\kappa(e)=1$, then it would be sampled in $G_s$ with probability $1$, and be approximately covered by $\vec{y},\vec{z}$. Thus, for any $e\in E^*$, $p(e)=\kappa(e)\cdot \rho_{\eps}$. Now, suppose $\kappa(E^*)>17\mu(G)\log n$. Then, observe that if $\vec{y},\vec{z}$ is output by \textsc{Static-Match}() (denote this event by $\mathcal{E}_{y,z}$) then none of the edges in $E^*$ were sampled (denote this event by $\mathcal{E}_2$). Thus, we have,
	\begin{align*}
	\prob{\mathcal{E}_{y,z}}\leqslant \prob{\mathcal{E}_2}&\leqslant \prod_{e\in E^*}(1-p(e))\leqslant \exp\paren{-\sum_{e\in E^*}\kappa(e)\cdot \rho_{\eps}}\leqslant \exp\paren{-17\cdot 2^{\nicefrac{1}{\eps^2}}\cdot \mu(G)\log n}
	\end{align*}
	Now, observe that if we are able to upper bound $\prob{\bigcup\limits_{\vec{y},\vec{z}\in \mathcal{D}}\mathcal{E}_{y,z}}$, then our lemma follows. This upper bound follows by taking a union bound over the set $\mathcal{D}$.
	\begin{align*}
	\prob{\bigcup\limits_{\vec{y},\vec{z}\in \mathcal{D}}\mathcal{E}_{y,z}}&\leqslant \exp\paren{-17\cdot 2^{\nicefrac{1}{\eps^2}}\cdot \mu\log n}\cdot 2^{\paren{\nicefrac{16}{\eps^2}}\cdot \mu\log n}=O\paren{\exp\paren{-2^{\nicefrac{1}{\eps^2}\cdot \mu \log n}}}
	\end{align*}
\end{proof}
\section{Algorithm \textsc{M-or-E*}()}
\label{sec:more}
In this section, we give the main subroutine \textsc{M-or-E*}() (see \Cref{lem:MorE}). We first define a few terms, and state a known result about finding an approximate fractional matching. 

\begin{definition2}
	Recall \Cref{def:lowedges}, and consider an integral matching $M$, define $E^M_{L}(G,\kappa)=E_{L}(G,\kappa)\cap M$, and let $V_{M}^{L}$ be the endpoints of $E_{L}^M(G,\kappa)$.  
\end{definition2}

\begin{lemma2}\label{lem:fracmatch}
		Given a multigraph $G$ (possibly non-bipartite), with edge capacities $\kappa$ and $\eps\in (0,1)$, such that $\kappa(D(e))\leqslant \nicefrac{1}{\alpha_{\eps}}$ for all $e\in E$, then there is an algorithm \textsc{Frac-Match}() that takes as input $G$, $\kappa$ and $\eps$, and returns a fractional matching $\vec{x}$ such that $\sum_{e\in E}x(e)\geqslant (1-\eps)\cdot \mu(G,\kappa)$, obeying the capacities $\kappa$ and the odd set constraints. The runtime of this algorithm is $O(\nicefrac{m\cdot \log n}{\eps})$.
\end{lemma2}

\begin{proof}
	For the case of bipartite graphs, there is an algorithm that takes as input a graph $G$, an edge capacity function $\kappa$ and $\eps\in (0,1)$, and returns in $O(\nicefrac{m}{\eps})$ time, a $(1+\eps)$-approximate fractional matching of $G$, obeying capacities $\kappa$ (see \cite{BPT20}). We use this algorithm for a non-bipartite graph as follows: we run the algorithm on $\textsc{bc}(G,\kappa)$, to obtain a matching $\vec{x}$. To get a fractional matching in $G$ that obeys $\kappa$, we do the following: let $e\in E(G)$ and let $e',e''\in E(\textsc{bc}(G))$ be copies of $e$. We let $z(e)=\frac{x(e')+x(e'')}{2}$. Then, the matching $\vec{z}$ obeys $\kappa$ as well as the fractional matching constraints since $\vec{x}$ obeys them.
\end{proof}
For the purpose of the algorithm recall, definition of $\kappa^+$ in \Cref{def:kappaplus} given $\kappa$ and $\eps$. We now formalize the algorithm \textsc{M-or-E*}(). Recall that it takes as input a multigraph $G$ with $\mu(G)\geqslant \nicefrac{\eps\cdot n}{16}$.

\begin{algorithm}
	\caption{\textsc{M-or-E*}($G,\kappa,\eps,\mu$)}
	\label{alg:MorE}
	\begin{algorithmic}[1]
		\State Include each $e\in E(G)$ independently with probability $p(e)=\kappa(e)\cdot \rho_{\eps}$ into graph $G_s$. 
		\State Let $M$ and $\vec{y},\vec{z}$ be the output of $\textsc{Static-Match}(G_s,\eps)$. \Comment{Phase 1}
		\If{$\card{M}<\mu-6\eps\mu$}\Comment{Phase 3}
		\State Return $E^*=\set{e\in E(G)\mid \kappa(e)<1\text{ and }yz(e)<1-\eps}$.
		\Else \Comment{Phase 2}
		\State $M_{I}\leftarrow M\setminus E_{L}(G,\kappa)$
		\State $\vec{y}\leftarrow M_{I}^D$\Comment{Converting $M_I$ into a matching on a multigraph.}
		\State\label{line:fracmatch}$\vec{x}\leftarrow \textsc{Frac-Match}(G[V_{L}^M]\cap E_{L}(G,\kappa),\kappa^+,\eps)$
		\EndIf
		\State Return $\vec{y}+ \vec{x}$.
	\end{algorithmic}
\end{algorithm}
We now show \Cref{lem:MorE} holds, we first restate it.

\lemtrio*

\begin{proof}[Proof of \Cref{lem:MorE}]
	We first show the runtime of the algorithm. Graph $G_s$ can be computed in time $O(m)$. Using \Cref{lem:propstatmatch}, we conclude that we can compute $E^*$ in order $O(\nicefrac{m}{\eps})$ time by running \textsc{Static-Match}($G_s,\eps$) and \Cref{lem:fracmatch} implies that \Cref{line:fracmatch} takes $O(\nicefrac{m\cdot \log n}{\eps})$ time. \\ \\
We show \Cref{lem:MorE}\ref{item:aa}. First observe that $V_{L}^M$ and $V(M_I)$ are disjoint, and since $\vec{y}$ and $\vec{x}$ are fractional matchings, $\vec{z}$ is also a fractional matching. Note that $\card{M}\geqslant \mu-7\eps\mu$. Moreover, $\sum_{e\in E}x(e)\geqslant (1-\eps)\cdot \mu(G[V_L^M]\cap E_{L},\kappa^+)$. This follows from applying \Cref{lem:fracmatch} to $G[V_L^M]\cap E_{L}(G,\kappa)$ with capacity function $\kappa^+$. Recall \Cref{def:kappaplus} and \Cref{def:lowedges} to see that $\kappa^+(D(e))\leqslant \nicefrac{1}{\alpha_{\eps}}$ for $e\in G[V_L^M]\cap E_{L}(G,\kappa)$, thus satisfying the requirements of \Cref{lem:fracmatch}. Next, applying \Cref{cor:sampling}, we have, $\mu(G_{s}[V_{L}^M]\cap E_{L})\leqslant \mu(G[V_{L}^M]\cap E_{L},\kappa^+)+5\cdot \eps\mu(G)$. Thus, we have $\sum_{e\in E}x(e)\geqslant (1-\eps)\cdot \paren{\mu(G_{s}[V_L^M]\cap E_L)-5\eps\mu}\geqslant (1-\eps)\cdot \paren{\card{M\setminus M_I}-5\eps\mu}$. This is because $M\setminus M_I$ is a matching of $G_{s}[V_{L}^M]\cap E_{L}$. So, $\sum_{e\in E}z(e)\geqslant \card{M_I}+(1-\eps)\cdot (\card{M\setminus M_I}-5\eps\mu)\geqslant (1-\eps)\cdot (\mu-12\eps\mu)\geqslant \mu-13\eps\mu$.\\ \\
We now show \Cref{lem:MorE}\ref{item:aa}\ref{item:ai} and \ref{item:aii}. Consider any edge $e\in \text{supp}(\vec{z})$ with $\kappa(D(e))\leqslant \nicefrac{1}{\alpha_{\eps}^2}$, $e\in G[V_{L}^M]\cap E_{L}(G,\kappa)$. Thus, $e\in \text{supp}(\vec{x})$, and therefore, from \Cref{lem:fracmatch}, $z(e)=x(e)\leqslant \kappa^+(e)\leqslant \kappa(e)\cdot \alpha_{\eps}$ and $z(D(e)) =x(D(e))\leqslant \kappa(D(e))\cdot \alpha_{\eps}$. Similarly, for any $e\in \text{supp}(\vec{z})$ with $\kappa(D(e))>\nicefrac{1}{\alpha_{\eps}^2}$, $e\in \text{supp}(\vec{y})$. By definition of $\vec{y}$, $z(e)=y(e)=\nicefrac{\kappa(e)}{\kappa(D(e))}$ (recall \Cref{def:dist}) and $z(D(e))=1$. This proves our claim. \\ \\
Next, we show \Cref{lem:MorE}\ref{item:bi}. First recall from the assumption of \Cref{lem:MorE} that $\mu\geqslant (1-\eps)\cdot \mu(G)$. From this fact and \Cref{lem:propE*}, we can conclude that $\kappa(E^*)=O(\mu\log n)$ and that for all $e\in E^*$, $\kappa(e)<1$. Next, observe that $\mu(G_s)\leqslant (1+\eps)\cdot \card{M}\leqslant (1-6\eps)\cdot \mu$. Applying \Cref{lem:guaran2} with $H=G_s$, we have, that for any matching $M'$ of $G$, $\card{E^*\cap M'}\geqslant \card{M'}-(1+\eps)^2\cdot (1-6\eps)\cdot \mu$. If $\card{M'}\geqslant (1-3\eps)\cdot \mu$, then we have $\card{E^*\cap M'}\geqslant \eps\cdot \mu$. Finally, consider any pair of vertices $u,v$ and let $e',e''\in D(u,v)$. Then, either both $e',e''$ are both approximately covered by $\vec{y},\vec{z}$ or neither of them are. This implies that either $D(u,v)\subseteq E^*$ or $D(u,v)\cap E^*=\emptyset$.
\end{proof}

\appendix{}
\section{Missing Proofs}
\label{sec:missingproof}
We start by giving a formal proof of the following lemma. 

\lemnono*

To prove the above lemma, we state the following observation.

\begin{observation2}\label{obs:folklore3}
	Suppose a fractional matching $\vec{x}$ of $G$ satisfies the odd set constraints for all odd sets of size smaller than $\nicefrac{3}{\eps}+1$, then, the fractional matching $\vec{z}=\frac{\vec{x}}{1+\eps}$ satisfies all odd set constraints.  
\end{observation2}
\begin{proof}
	Suppose $\vec{x}$ is a fractional matching of $G$ that satisfies odd set constraints for all odd sets of size smaller than $\nicefrac{3}{\eps}+1$. Let $\vec{z}$ be a fractional matching obtained by scaling down $\vec{x}$ by $1+\eps$. Now, consider any odd set $B$ with $\card{B}\leqslant \nicefrac{3}{\eps}$, then $\vec{x}$, and therefore $\vec{z}$ satisfies the odd set constraint corresponding to this. So, we consider $B$ such that $\card{B}\geqslant \nicefrac{3}{\eps}+1$, then, 
	\begin{align*}
	\frac{\card{B}}{\card{B}-1}&=1+\frac{1}{\card{B}-1}\leqslant 1+\nicefrac{\eps}{3}
	\end{align*}
	The last inequality follows because of the fact that $\card{B}-1\geqslant \nicefrac{3}{\eps}$. Since $\vec{x}$ satisfies the fractional matching constraints, this implies: $\sum_{e\in G[B]} x(e)\leqslant \frac{\card{B}}{2}$. Thus, we have, $\sum_{e\in G[B]} z(e)\leqslant \sum_{e\in G[B]}\frac{x(e)}{1+\eps}\leqslant \frac{\card{B}}{2\cdot (1+\eps)}$. By the above argument, this is upper bounded by $\frac{\card{B}-1}{2}$. Thus, $\vec{z}$ satisfies all fractional matching constraints as well as the odd set constraints.
\end{proof}

We now state the following lemma. 

\begin{proof}[Proof of \Cref{lem:folklore2}]
Consider the maximum fractional matching $\vec{x}$ of $\textsc{bc}(G,\kappa)$, that obeys the capacity constraints $\kappa$. Let $\vec{z}$ be a fractional matching of $(G,\kappa)$ that is obtained from $\vec{x}$ as follows. Let $e\in E(G)$, and let $e'$, $e''$ be copies of $e$ in $\textsc{bc}(G,\kappa)$. Then, we let $z(e)=\frac{x(e')+x(e'')}{2}$. Note that $\vec{z}$ obeys fractional matching constraints, since $\vec{x}$ does. Additionally, observe that $z(e)\leqslant \kappa(e)$. Thus, for any odd set $B\subset V$ with $\card{B}\leqslant \nicefrac{3}{\eps}$, we have,
	\begin{align*}
	\sum_{e\in G[B]}z(e)&=\sum_{u,v\in B}z(D(u,v))\\
	&\leqslant \sum_{u,v\in B}\kappa(D(u,v))\\
	&\leqslant \nicefrac{1}{\alpha_{\eps}}\\
	&\text{(Since }\kappa(D(e))\leqslant \nicefrac{1}{\alpha_{\eps}})\\
	&\leqslant \sum_{u,v\in B}\nicefrac{\eps}{3}\\
	&\text{(Since }\nicefrac{1}{\alpha_{\eps}}\leqslant \nicefrac{\eps}{3} \text{ for $\eps<\nicefrac{1}{2}$)}\\
	&\leqslant \frac{\eps}{3}\cdot\frac{\card{B}\cdot \paren{\card{B}-1}}{2}\\
	&\leqslant \frac{\card{B}-1}{2}\\
	&\text{(Since }\card{B}\leqslant \nicefrac{3}{\eps})
	\end{align*}
	Thus, $\vec{z}$ satisfies the small odd set constraints, and we know from Observation \ref{obs:folklore3} that $\vec{y}=\frac{\vec{z}}{1+\eps}$ satisfies all odd set constraints in addition to the fractional matching constraints. Thus, we have, $(1+\eps)\cdot \sum_{e\in E(G)}y(e)=\sum_{e\in E(G)} z(e)=0.5\sum_{e\in \textsc{bc}(G)}x(e)$. Thus, if $\vec{x}$ is the optimal fractional matching of $\textsc{bc}(G)$, then we have, 
	\begin{align*}
	\mu(G,\kappa)\cdot (1+\eps)&\geqslant (1+\eps)\cdot \sum_{e\in E(G)}y(e)\\
	&=0.5\sum_{e\in \textsc{bc}(G)}x(e)\\
	&\text{(From the discussion above)}\\
	&=0.5\cdot \mu(\textsc{bc}(G), \kappa)\\
	&\text{(Since }\vec{x}\text{ is the maximum fractional matching of }(G,\kappa))
	\end{align*}
	This proves our claim. 
\end{proof}

\section{Reduction from Simple to Multigraphs}

In this section, we will prove the following lemma.

\lemmreduc*

In order to do that, we first start by proving the following simple observation.

\begin{observation2}\label{obs:nummatch}
	Let $G$ be any simple graph with maximum matching size $\mu(G)$, then $G$ contains at most $2^{O(\mu(G)\cdot \log n)}$ matchings of any size.
\end{observation2}
\begin{proof}
	Since the total number of edges is at most $n^2$, there are at most $n^2\choose j$ ways of choosing a matching of size $j$. Thus, we have, the total number of matchings possible is at most:
	\begin{align*}
	\sum_{j=1}^{\mu(G)} {n^2\choose j}\leqslant \sum_{j=1}^{\mu(G)}{2^{2\cdot j\log n}}\leqslant \mu(G)\cdot 2^{2\cdot \mu(G)\log n}\leqslant 2^{3\cdot \mu(G)\log n}
	\end{align*}
	This proves our claim.
\end{proof}
Now, in order to prove \Cref{lem:reduction}, we first give a procedure that takes as input a simple graph, and outputs a multigraph with $O(\nicefrac{\mu(G)}{\eps})$ vertices that preserves a fixed matching $M$ of $G$ approximately with probability $1-\exp\paren{-O\paren{\card{M}\cdot\eps^3}}$.

\begin{algorithm}[H]
	\algorithmicrequire{ Graph $G$, a sparsification parameter $\tau$, and a parameter $\eps>0$}\\
	\algorithmicensure{ A multigraph $\mathcal{G}(\mathcal{V},\mathcal{E})$ with $\card{\mathcal{V}}=\tau$}
	\caption{\textsc{Vertex-Red-Basic}($G,\tau,\eps$)}
	\begin{algorithmic}[1]
		\State Partition $V$ into $\tau$ bins, $\mathcal{V}\coloneqq (B_1,B_2\cdots, B_{\tau})$, by assigning every vertex to one of the $\tau$ bins uniformly at random. 
		\State For a vertex $u$, let $B(u)$ denote the bin chosen for $u$. For any edge $(u,v)\in E$, if $B(u)\neq B(v)$, then add an edge $e_{u,v}$ between $B(u)$ and $B(v)$.
		\State Return the multigraph $\mathcal{G}(\mathcal{V},\mathcal{E})$.
	\end{algorithmic}
	\label{alg:vertexred}
\end{algorithm}

\begin{lemma2}\label{lem:subred2}
	Let $G$ be a simple graph, let $\eps\in (0,1)$, and let $\tau\geqslant \frac{4\cdot \mu(G)}{\eps}$. Let $\mathcal{G}=\textsc{Vertex-Red-Basic}(G,\tau,\eps)$, and let $M$ be any fixed matching of $G$. Then with probability $1-2^{-\frac{\card{M}\cdot \eps^3}{32}}$, there exists a matching $\mathcal{M}$ in $\mathcal{G}$ such that if $\mathcal{M}\subset M$, and $\card{\mathcal{M}}\geqslant (1-\eps)\cdot \card{M}$.
\end{lemma2}
\begin{proof}
	Let $\delta=\nicefrac{\eps}{4}$, and we define $t\coloneqq 2\card{M}$. Since $\eps<1$, this implies that $\delta<\nicefrac{1}{2}$. Recall that we have $\tau=\frac{4\mu(G)}{\eps}$ bins, and we combine these bins arbitrarily so that we end up with $\frac{t}{\delta}$ groups. Call these groups $Z_1,\cdots, Z_{\nicefrac{t}{\delta}}$. Note that every vertex $v$ lands in bin $Z_i$ with probability $\nicefrac{\delta}{t}$. Call a group $Z_i$ bad if it doesn't contain even one of the $2\card{M}$ vertices of $M$. Associate a random variable $\textsf{X}_i$ with $Z_{i}$ that takes value $1$ if $Z_i$ is bad, and zero otherwise. Let $\textsf{X}=\sum_{i=1}^{\nicefrac{t}{\delta}}\textsf{X}_i$ denote the number of bad groups. We have, 
	\begin{align*}
	\prob{\textsf{X}_i=1}=\paren{1-\frac{\delta}{t}}^t=e^{-\delta}\leqslant 1-\delta+\nicefrac{\delta^2}{2}
	\end{align*}
	Thus, we have, $\expect{\textsf{X}}\leqslant \nicefrac{t}{\delta}\paren{1-\delta+\nicefrac{\delta^2}{2}}$. Note that $\textsf{X}$ is a sum of negatively correlated variables. This is because if $Z_i$ is empty, then $Z_j$ has an increased likelihood of being not empty for $i\neq j$. So, using \Cref{lem:chernoff}, we have,
	\begin{align*}
	\prob{\textsf{X}\geqslant \paren{1+\delta^2}\cdot \frac{t}{\delta}\cdot \paren{1-\delta+\nicefrac{\delta^2}{2}}}&\leqslant \exp\paren{-\delta^4\cdot \frac{t}{\delta}\cdot \paren{1-\delta+\frac{\delta^2}{2}}}\\
	&\leqslant \exp\paren{-\delta^3\cdot t+\delta^4\cdot t}\\
	&\leqslant \exp\paren{\frac{-\delta^3}{2}\cdot t}\\
	&\text{(Since $\delta<\frac{1}{2}$)}
	\end{align*}
	Thus, with probability at least $1-2^{-\frac{\card{M}\cdot \eps^3}{128}}$, $\textsf{X}\leqslant \paren{\frac{t}{\delta}-t+2\delta t}$. Thus, with probability at least $1-2^{-\frac{\card{M}\cdot \eps^3}{128}}$, we have that $t-2\delta \cdot t$ groups are good, which means that they contain at least one vertex of $M$. Now, we need to show that if at least $t-2\delta \cdot t$ bins are good, then the edges of $M$ form a matching $\mathcal{M}$ of $\mathcal{G}$ such that $\card{\mathcal{M}}\geqslant \card{M}\cdot(1-\eps)$. For every bin, we fix one vertex of $M$, and remove the rest. Since $t-2\delta \cdot t$ of the bins are good, this implies that we lost at most $2\delta \cdot t$ vertices of the matching $M$. Consequently, we deleted at most $2\delta \cdot t$ edges of $M$. The remaining edges have their endpoints in different bins, and therefore, they form a matching in $\mathcal{G}$. The size of this matching is $\card{M}-2\delta \cdot t$, which is equal to $\card{M}-\eps \card{M}$. Thus, we have our claim. 
\end{proof}

\begin{lemma2}\label{lem:subred1}
	Let \textsc{Vertex-Red}() be an algorithm that does independent runs of \textsc{Vertex-Red-Basic}(). Let $H_1,H_2\cdots, H_{\lambda}$ be multigraphs on independent runs of \textsc{Vertex-Red-Basic}() on input $G,\tau\geqslant \frac{4\cdot \mu(G)}{\eps}$ and $\eps\in (0,\nicefrac{1}{2})$, for $\lambda \geqslant \frac{1600\cdot \log n}{\eps^4(1-\eps)}$. Then, with probability at least $1-\exp(-O(\nicefrac{\mu(G)\cdot \log n}{\eps}))$, for every matching $M$ with $\card{M}\geqslant (1-\eps)\cdot \mu(G)$, there is a matching $M'$ in some $H_i$ such that $M'\subset M$, and $\card{M'}\geqslant (1-\eps)\card{M}$. 
\end{lemma2}
\begin{proof}
	Consider a fixed matching $M$, of size at least $(1-\eps)\cdot \mu(G)$, then with probability at least $1-\exp\paren{-\frac{\mu(G)\cdot (1-\eps)\cdot \eps^3}{128}}$, a fixed $H_i$ contains a matching $M'\subset M$ with $\card{M'}\geqslant (1-\eps)\card{M}$. This is implied by \Cref{lem:subred2}. Since each of the $H_j$'s are independent runs of $\textsc{Vertex-Red}()$, with probability at least $1-\exp\paren{-\frac{12\cdot\mu(G)\cdot \log n}{\eps}}$, some $H_i$ contains a matching $M'\subset M$ with $\card{M'}\geqslant (1-\eps)\card{M}$. Taking a union bound over all possible matchings of size at least $(1-\eps)\cdot \mu(G)$, from Observation \ref{obs:nummatch}, we have that with probability at least $1-\exp\paren{-\frac{9\cdot \mu(G)\cdot \log n}{\eps}}$, for every matching $M$ with $\card{M}\geqslant (1-\eps)\cdot \mu(G)$, there is a matching $M'$ in some $H_i$ such that $M'\subset M$ and $\card{M'}\geqslant (1-\eps)\cdot \card{M}$.
\end{proof}

We now show the following. 

\begin{lemma2}\label{lem:mainvertexred}
	Let $G$ be a simple graph, and let $H_1,\cdots,H_{\lambda}$ be the output of independent runs of \textsc{Static-Match}() with $G$ and $\eps$ as input, where $\lambda\geqslant \frac{1600\cdot \log n}{\eps^4\cdot (1-\eps)}$. Consider an adversary that deletes edges from $G$. We simulate this deletion process in $H_i$'s, that is, if $(u,v)$ is deleted from $G$, then $e_{u,v}$ (if present) is deleted from each of the $H_i$'s. Let $\mu_{i}^t$ denote the size of the maximum matching of $H_i$ after $t$ deletions, and similarly, let $\mu^{t}$ denote the size of the maximum matching of $G$ after $t$ deletions. Suppose $\mu^t>(1-\eps)\mu^0$, then, $\mu^t_{i}>(1-2\eps)\mu^0$ for some $i\in [\lambda]$.
\end{lemma2}
\begin{proof}
	Let $M$ be a matching in $G$ satisfying the statement of the lemma. That is, suppose at time $t$, $\card{M}>(1-\eps)\mu^0$. Since the adversary is only performing deletions, $\card{M}>(1-\eps)\mu^0$ initially, when no deletions had been performed. So, by \Cref{lem:subred1}, before any deletions are performed, there is a matching $M'$ in one of the $H_i$'s such that $M'\subset M$, and $\card{M'}\geqslant (1-\eps)\cdot \card{M}>(1-2\eps+\eps^2)\mu^0$. Since all of the edges of $M$ survive at time $t$, this is true of $M'$ as well. Thus, $\mu^t_{i}>(1-2\eps)\mu^0$ for some $i\in [\lambda]$. This proves our claim. 
\end{proof}

We now show that \Cref{lem:reduction} follows from the lemmas proved in this section.

\begin{proof}[Proof of \Cref{lem:reduction}]
	Consider \Cref{lem:reduction}\ref{item:first}. This is implied by the properties of \textsc{Vertex-Red}(), \Cref{lem:subred1}, and \Cref{lem:subred2}. Property \Cref{lem:reduction}\ref{item:third} is on the other hand implied by the contrapositive of \Cref{lem:mainvertexred}. 
\end{proof}

\section{Properties of \textsc{Static-Match}()}
\label{app:statmatch}

In this section we will show \Cref{lem:propstatmatch}, which we first restate.

\statmatch*
We first note that the properties \ref{item:a}-\ref{item:b} and \ref{item:d}-\ref{item:folklore} are evident in Property 3.1 and the proof of Lemma 2.3 of \cite{DP14}. We now show that \ref{item:c} also holds.

\begin{proof}[Proof of \Cref{lem:propstatmatch}]
	Let $\textsc{Basic-Static-Match}()$ be the algorithm satisfying of \cite{DP14} satisfying \ref{item:a}-\ref{item:b} and \ref{item:d}-\ref{item:folklore}. We now show how to obtain \textsc{Static-Match}() from \textsc{Basic-Static-Match}(). We run \textsc{Basic-Static-Match}() with input $G$ and $\delta=\nicefrac{\eps}{3}$. Thus, we get a matching $M$ of size at least $(1-\nicefrac{\eps}{3})\cdot \mu(G)$. Moreover, the duals $\vec{y}$ and $\vec{z}$ output by the algorithm have $yz(V)\leqslant (1+\nicefrac{\eps}{3})\cdot \mu(G)$. We consider the following procedure: for every $B\in \Omega$ with $z(B)>0$ and $\card{B}\geqslant \frac{3}{\eps}+1$, we increase $y(v)$ by $\nicefrac{z(B)}{2}$ for every $v\in B$, and we decrease $z(B)$ to $0$. Thus, each $y(v)$ is still a multiple of $\eps$ and each $z(B)$ is still a multiple of $\eps$. We conclude that \ref{item:d} still holds. Moreover, this transformation keeps the value of the dual constraint for every edge the same. Thus, we still satisfy \ref{item:e}. We update $\Omega$ by removing $B$ from it. The set $\Omega$ still remains laminar. Thus, \ref{item:b} is still satisfied. We first observe that the time taken to do this is linear in the sum of the sizes of the odd sets $B$ with $z(B)>0$. Note that since $yz(V)$ sums to at most $(1+\nicefrac{\eps}{3})\mu(G)$, and non-zero $z(B)$ have value at least $\nicefrac{\eps}{3}$, this implies that the sums of the sizes of the odd sets is at most $\nicefrac{3}{\eps}\cdot (1+\nicefrac{\eps}{3})\cdot \mu(G)$. Thus, the time taken for the procedure is $O(\nicefrac{m}{\eps})$. Next, we observe that the value of $yz(V)$ changes by at most $\sum_{B:\card{B}\geqslant \nicefrac{3}{\eps}+1}\frac{z(B)}{2}$. This value is at most $\nicefrac{\eps}{3}\cdot (1+\nicefrac{\eps}{3})\cdot \mu(G)$, as shown by the following calculation.
	
	\begin{equation*}
	\begin{split}
		\paren{\frac{3}{\eps}}\cdot \sum_{B:\card{B}\geqslant \nicefrac{3}{\eps}+1}\frac{z(B)}{2}\leqslant \sum_{B:\card{B}\geqslant \nicefrac{3}{\eps}+1}\paren{\frac{\card{B}-1}{2}}\cdot z(B)&\leqslant (1+\nicefrac{\eps}{3})\cdot \mu(G)
	\end{split}
	\end{equation*}
	
	Thus, the new $yz(V)$ has value at most $(1+\nicefrac{\eps}{3})^2\cdot \mu(G)$. This implies that \ref{item:folklore} holds. Finally, since every blossom $B$ of size at least $\nicefrac{3}{\eps}+1$ has $z(B)=0$, thus, \ref{item:c} holds. 
\end{proof}

\section{Rounding Fractional Matchings}
\label{sec:rounding}

In this section, we will prove \Cref{thm:sparse}. More concretely, we give state the procedure \textsc{Sparsification}() that takes as input a fractional matching $\vec{x}$ of a simple graph $G$, and outputs a graph $H$, of size $\Tilde{O}(\mu(G))$. If $x(e)\leqslant \eps^6$ for all $e\in E(G)$, then, $H$ contains an integral matching of size at least $(1-10\eps)\cdot \sum_{e\in E}x(e)$ in its support. The proof of this theorem is implicit in the work of \cite{Wajc2020}, but we describe it here for completeness. We first state the algorithm, and then describe some of its properties. The algorithm uses a dynamic edge coloring algorithm as a subroutine. The update time of the subroutine is $O(\log n)$ in the worst case. The sparsification algorithm takes as input a fractional matching $\vec{x}$ and a parameter $\eps>0$. Then, it classifies $E(G)$ into classes as follows: $E_i=\set{e\mid x(e)\in [(1+\eps)^{-i}, (1+\eps)^{-i+1})}$. Note that we only consider edges $e\in E(G)$ with $x(e)\geqslant \paren{\nicefrac{\eps}{n}}^2$, since the total contribution of these edges to fractional matching is at most $\eps^2$, and we can afford to ignore them if we want to compute a $(1+\eps)$ approximation. We now state the algorithm. 

\begin{algorithm}
	\caption{\textsc{Sparsification}($\vec{x},\eps$)}
	\begin{algorithmic}[1]
		\State $d\leftarrow \frac{4\cdot \log\paren{\nicefrac{2}{\eps}}}{\eps^2}$
		\For{$i\in \set{1,2,\cdots, 2\log_{1+\eps}(\nicefrac{n}{\eps})}$}
		\State Compute a $3\lceil (1+\eps)^i\rceil$-edge colouring $\Phi_i$ of $E_i$.
		\State Let $S_i$ be a sample of $3\cdot \min\set{\lceil d\rceil,\lceil (1+\eps)^i\rceil}$ colours without replacement in $\Phi_i$. \label{line:blah}
		\State Return $K=(V,\cup_{i}\cup_{M\in S_i}M)$
		\EndFor
	\end{algorithmic}
\end{algorithm}

We state the guarantees of the edge coloring subroutine. 

\begin{observation2}
	The size of $K$ output by \textsc{Sparsification}($\vec{x},\eps$) is,
	\begin{align*}
	\card{E(H)}=O\paren{\frac{\log (\nicefrac{n}{\eps})}{\eps}\cdot d\cdot \mu(\textup{supp}(H))}
	\end{align*}
\end{observation2}

\begin{lemma2}\label{lem:edgecolour}\cite{BDHN18}
There is a deterministic dynamic algorithm that maintains a $2\Delta-1$ edge coloring of a graph $G$ in $O(\log n)$ worst case update time, where $\Delta$ is the maximum degree of the graph $G$. 
\end{lemma2}

\begin{observation2}\label{obs:sampling}
Let $\eps\in (0,\nicefrac{1}{2})$ and suppose the input to \textsc{Sparsification}($\vec{x},\eps$) is a matching $\vec{x}$ with $x(e)\leqslant \eps^6$, then, $\card{S_i}=3\cdot \lceil d\rceil$. 
\end{observation2}

To show that $K$ contains a matching of size at least $(1-\eps)\cdot \sum_{e\in E(G)}x(e)$ in its support, we will show a random fractional matching $\vec{y}$ in $K$ that sends flow at most $\eps$ through each of its edges. Moreover, $\expect{\sum_{e\in E}y(e)}\geqslant (1-6\eps)\cdot \sum_{e\in E(G)} x(e)$. This will show the existence of a large matching that sends flow at most $\eps$ through each of its edges. To show this, we give some properties of the algorithm.

\begin{lemma2}\label{lem:probbounds}
Let $\eps\in (0,\nicefrac{1}{2})$ and let $\vec{x}$ be the input to \textsc{Sparsification}($\vec{x},\eps$) such that $x(e)\leqslant \eps^6$ for all $e\in E(G)$. Then, for every edge $e$, $\prob{e\in K}\in \bracket{\nicefrac{x(e)\cdot d}{(1+\eps)^2},x(e)\cdot d\cdot (1+\eps)}$.
\end{lemma2}
\begin{proof}
Since $\vec{x}$ has the property that for every $e\in E$, $x(e)\leqslant \eps^6$, this implies that we sample $3\cdot \lceil d\rceil$ from $\Phi_i$ colors for each $i$ (from Observation \ref{obs:sampling}). Thus, if we consider an edge $e\in E_i$, then, we have,
\begin{align*}
\prob{e\in K}=\frac{\lceil d\rceil}{\lceil(1+\eps)^i\rceil}\leqslant \frac{d\cdot (1+\eps)}{(1+\eps)^i}\leqslant d\cdot (1+\eps)\cdot x(e)
\end{align*}
The last inequality follows from the fact that $x(e)\in \bracket{(1+\eps)^{-i},(1+\eps)^{-i+1}}$. Moreover, we have,
\begin{align*}
	\frac{\lceil d\rceil}{\lceil (1+\eps)^i\rceil}\geqslant \frac{d}{(1+\eps)^i+1}
	\geqslant \frac{d}{(1+\eps)^{i+1}}
	\geqslant \frac{d\cdot x(e)}{(1+\eps)^2}
	\end{align*}
The first inequality follows from the fact that $\lceil d\rceil \geqslant d$, and $\lceil (1+\eps)^i\rceil\leqslant (1+\eps)^i+1$. The second inequality follows from the following reasoning.
\begin{align*}
(1+\eps)^{-i}\leqslant x(e)\leqslant \eps^6\leqslant \nicefrac{1}{d}\leqslant \eps
\end{align*}
Thus, $\nicefrac{1}{\eps}\leqslant (1+\eps)^i$, and $1\leqslant \eps\cdot (1+\eps)^i$, so the second inequality follows. The last inequality follows from the fact that $x(e)\in \bracket{(1+\eps)^{-i},(1+\eps)^{-i+1}}$.
\end{proof}

For an edge $e\in E(G)$, we define $X_e$ to be an indicator random variable that takes value $1$ if $e\in K$, and $0$ otherwise. 

\begin{observation2}\label{obs:negatcorr}
	For a vertex $v$, the variables $\set{X_e\mid e\text{ incident on }v}$ are negatively associated. 
\end{observation2}
\begin{proof}
	If the edges $e$ and $e'$ incident on $v$ belong in $E_i$ and $E_j$ where $i\neq j$, then $X_e$ and $X_{e'}$ are independent. On the other hand, if they belong in the same $E_i$, then recall we did an edge colouring on $E_i$, so $e'$ and $e$ got different colours. So, if the colour corresponding to $e$ is picked into $K$, then this reduces the probability of the colour corresponding to $e'$ being picked into $K$.
\end{proof}

\begin{lemma2}\label{cor:cond}
	From Observation \ref{obs:negatcorr} and \Cref{lem:probbounds}, we can conclude that for edges $e$ and $e'$ incident on $v$, 
	\begin{align*}
	\prob{X_e=1\mid X_{e'}=1}\leqslant \prob{X_e=1}\leqslant x(e)\cdot d\cdot (1+\eps)
	\end{align*}
\end{lemma2}

\begin{lemma2}
	For any vertex $v$ and edge $e'$ incident on $v$, the variables $\set{[X_e\mid X_{e'}]\mid e\text{ incident on }v}$ are negatively associated.
\end{lemma2}

\begin{lemma2}[Concentration Bound]\label{lem:concentrate}
	Let $X$ be the sum of negatively associated random variables $X_1,\cdots, X_k$ with $X_i\in [0,M]$ for each $i\in [k]$. Then, for $\sigma^2=\sum_{i=1}^k \var{X_i}$ and all $a>0$, 
	\begin{align*}
	\prob{X>\expect{X}+a}\leqslant \exp\paren{\frac{-a^2}{2\paren{\sigma^2+\nicefrac{a\cdot M}{3}}}}
\end{align*}
\end{lemma2}

Following theorem directly implies \Cref{thm:sparse}. 

\begin{theorem2}\label{thm:largeintegral}
Let $K$ be the subgraph of $G$ output by \textsc{Sparsification}(), when run on the matching $\vec{x}$ with parameters $\eps\in (0,\nicefrac{1}{2})$. If $x(e)\leqslant \eps^6$ for all $e\in E$, then $K$ supports a fractional matching $\vec{y}$ such that $y(e)\leqslant \eps$, and $\sum_{e\in E}y(e)\geqslant (1-6\eps)\cdot \sum_{e\in E}x(e)$.
\end{theorem2}
\begin{proof}
	To show the theorem, we will describe a process that simulates \textsc{Sparsification}($\vec{x},\eps$), and outputs a random matching $\vec{y}$ such that $y(e)\leqslant \eps$ for every $e\in E$. Additionally, $\expect{\sum_{e\in E}y(e)}\geq(1-6\eps)\sum_{e\in E}x(e)$. As an intermediate step, let $z(e)=\frac{(1-4\eps)}{d}\cdot X_e$. So, we have,
	\begin{align*}
	\expect{z(e)}&\geqslant \expect{z(e)\mid X_e=1}\cdot \prob{X_e=1}\geqslant \frac{1-4\eps}{d}\cdot \frac{x(e)\cdot d}{(1+\eps)^2}\geqslant (1-6\eps)\cdot x(e)
	\end{align*}
	
The second to last inequality follows from \Cref{lem:probbounds}. Next define $\vec{y}$ as follows:
\begin{align*}
y(e)=\begin{cases}
0 \text{ if for one of the endpoints }v\text{ of }e,\  \sum_{e'\in v}z(e')>1\\
z(e) \text{ otherwise.}
\end{cases}
\end{align*}
Essentially, our procedure is creating matching $\vec{z}$ as follows: it assigns value $\frac{1-4\eps}{d}$ to $z(e)$ if $e$ was included in $K$, and $0$ otherwise. The value $y(e)$ is the same $z(e)$ in all cases except when $e$ is picked into $K$, but at one of the endpoints $v$, $\sum_{e'\ni v}z(e')>1$ , that is, the fractional matching constraint is violated at $v$. We show that it is unlikely that an edge (when picked) has one of its endpoints violated. Let $e'$ be an edge incident on $v$, we to bound the probability that $z(e')\neq y(e')$. 
\begin{align*}
\expect{\sum_{e\in v}z(e)\mid X_{e'}=1}&\leqslant \frac{1-4\eps}{d}+\expect{\sum_{e'\neq e, e\in v}z(e)\mid X_{e'}}\\
&\leqslant \eps+\sum_{e\in v}x(e)\cdot (1+\eps)\cdot d\cdot\paren{\frac{1-4\eps}{d}}\\
&\text{(Since $\nicefrac{1}{d}\leqslant \eps$ and from \Cref{cor:cond})}\\
&\leqslant (1-\eps)
\end{align*}
Note that at any end point of $e'$, conditioned on $e'$ being sampled, the expected sum of $z(e)$'s at that endpoint is upper bounded by $(1-\eps)$. Now, $z(e)$ is assigned value $0$ only if the sum of the $z(e)$'s deviates from the expected value by $\eps$. To see this, we want to compute $\var{[z(e)\mid X_{e'}]}$, where $e$ and $e'$ share an end point. Note that $[z(e)\mid X_{e'}]$ takes value $\frac{1-4\eps}{d}$ with probability $\prob{X_e\mid X_{e'}}$, otherwise it takes value $0$.
\begin{align*}
\var{\bracket{z(e)\mid X_{e'}}}&\leqslant \expect{\bracket{z(e)\mid X_{e'}}^2}\\
&\leqslant\paren{\frac{1-4\eps}{d}}^2\cdot \prob{X_e\mid X_{e'}}\\
&\leqslant \paren{\frac{1-4\eps}{d}}^2\cdot x(e)\cdot d\cdot (1+\eps)\\
&\text{(From \Cref{cor:cond})}\\
&\leqslant \frac{x(e)}{d}
\end{align*}
This implies that $\sum_{e\in v}\var{\bracket{z(e)\mid X_{e'}}}\leqslant \frac{1}{d}$. So, we want to compute the probability that the sum of the random variables $\set{\bracket{z(e)\mid X_{e'}}}$ deviates from the expected value by $\eps$. Applying \Cref{lem:concentrate}, we have,
\begin{align*}
\prob{\sum_{e\in v}\bracket{z(e)\mid X_{e'}}\geqslant \expect{\sum_{e\in v}\bracket{z(e)\mid X_{e'}}}+\eps}&\leqslant \exp\paren{\frac{-\eps^2}{2\paren{\frac{1}{d}+\frac{\eps}{3\cdot d}}}}\\
&\paren{\text{Since }z(e)\in \bracket{0,\nicefrac{1-4\eps}{d}}}\\
&\leqslant \exp\paren{-\eps^2\cdot 0.25\cdot d}\\
&\text{(Since }\eps\in (0,1))\\
&\leqslant \frac{\eps}{2}\\
&\text{(Since }d=\frac{4\cdot \log(\nicefrac{2}{\eps})}{\eps^2})
\end{align*}
Taking union bound over both endpoints, we know that $\prob{y(e)=z(e)\mid X_e=1}\geqslant (1-\eps)$. Thus, we have:
\begin{align*}
\expect{y(e)}&=\paren{\frac{1-4\eps}{d}}\cdot \prob{y(e)=z(e)}\\
					&=\paren{\frac{1-4\eps}{d}}\cdot \prob{y(e)=z(e)\mid X_e=1}\prob{X_e=1}\\
					&\geqslant \paren{\frac{1-4\eps}{d}}\cdot (1-\eps)\cdot \frac{x(e)\cdot d}{(1+\eps)^2}\\
					&\geqslant (1-7\eps)\cdot x(e)
\end{align*}
Thus, $\expect{\sum_{e\in E}y(e)}\geqslant (1-7\eps)\cdot \sum_{e\in E}x(e)$. Moreover, for all $e\in E$, $y(e)\leqslant \eps^6$. This proves our claim. 
\end{proof}

We now restate the lemma, and then show its proof. 

\lemuno*
\begin{proof}
Note that \ref{lem:sparseb} is implied by \Cref{thm:largeintegral}, and the fact that any fractional matching $\vec{y}$ with $y(e)\leqslant \eps$ for all $e\in E$, satisfies odd set constraints for all odd sets of size at most $\nicefrac{1}{\eps}$. To see \ref{lem:sparsea}, note that deleting $e$ from supp($\vec{x}$) corresponds to just deleting $e$ from $G_i$, and since our edge coloring algorithm is able to handle edge insertions and deletions in $\Tilde{O}_{\eps}(1)$ time, such updates can be handled in $\Tilde{O}_{\eps}(1)$ update time (see \Cref{lem:edgecolour}). Finally, if an update reduces $x(e)$ for some $e\in E$, then this corresponds to deleting $e$ from some $G_i$ and adding it to $G_j$ for some $j<i$. Thus, the edge colouring algorithms running on $G_i$ and $G_j$ have to handle an edge insertion and deletion respectively, and this can be done in $\Tilde{O}_{\eps}(1)$ time (see \Cref{lem:edgecolour}). 
\end{proof}

\newpage
\bibliographystyle{alpha}
\bibliography{references.bib}
\end{document}